\documentclass[11pt]{article}
\usepackage[margin=1in]{geometry}

\usepackage{amsmath}
\usepackage{amsthm}
\usepackage{bbm}
\usepackage{amsfonts}
\usepackage{xcolor}
\usepackage{bbm}
\usepackage{cancel}
\usepackage{authblk}
\definecolor{darkblue}{rgb}{0, 0, 0.5}
\usepackage[colorlinks=true,breaklinks=true,bookmarks=true,
urlcolor=darkblue,citecolor=darkblue,linkcolor=darkblue,
bookmarksopen=false,draft=false,pagebackref=true]{hyperref}
\renewcommand*\backref[1]{\ifx#1\relax \else (Cited on #1) \fi}
\usepackage{mathtools}
\usepackage{tikz}
\usepackage{bm}
\usetikzlibrary{arrows.meta}
\usepackage{subcaption}
\usepackage{natbib}

\newtheorem{theorem}{Theorem}
\newtheorem*{theorem*}{Theorem}
\newtheorem{proposition}{Proposition}

\newtheorem{lemma}{Lemma}

\newtheorem{observation}{Observation}

\theoremstyle{definition}
\newtheorem{example}{Example}
\newtheorem{definition}{Definition}
\theoremstyle{definition}
\newenvironment{examplecon}[1]
  {\par\bigskip\noindent\textbf{Example #1 (continued). }}
  {\par}

\usepackage[colorinlistoftodos,textsize=tiny]{todonotes}
\newcommand{\Comments}{0}
\definecolor{gray}{gray}{0.5}
\definecolor{darkgreen}{rgb}{0,0.5,0}
\definecolor{myorange}{rgb}{0.2, 0.3, 0.0}
\newcommand{\mynote}[2]{\ifnum\Comments=1\textcolor{#1}{#2}\fi}
\newcommand{\mytodo}[2]{\ifnum\Comments=1\todo[linecolor=#1!80!black,backgroundcolor=#1,bordercolor=#1!80!black]{#2}\fi}
\newcommand{\raf}[1]{\mynote{darkgreen}{{[RF: #1]}}}

\newcommand{\bo}[1]{{\mynote{blue}{[Bo: #1]}}}

\newcommand{\mary}[1]{\mynote{purple}{{[MM: #1]}}}
\newcommand{\maryt}[1]{\mytodo{purple!20!white}{MM: #1}}

\newcommand{\Var}{\mathop{\mathrm{Var}}}

\newcommand{\D}{\mathcal{D}}
\newcommand{\E}{\mathbb{E}}

\renewcommand{\P}{\mathcal{P}}
\newcommand{\Q}{\mathcal{Q}}
\newcommand{\R}{\mathcal{R}}
\newcommand{\T}{\mathcal{T}}

\newcommand{\X}{\mathcal{X}}

\newcommand{\ones}{\mathbbm{1}}

\def\reals{\mathbb{R}}

\newcommand{\cS}{\mathcal{S}}
\newcommand{\cM}{\mathcal{M}}

\newcommand{\Lorder}{\mathcal{A}_d^{\text{lin}}}
\newcommand{\Aquery}{\mathcal{A}_d}

\DeclareMathOperator{\arcosh}{arcosh}

\newcommand{\orderc}{\mathsf{order\text{-}c}}
\newcommand{\queryc}{\mathsf{query\text{-}c}}
\newcommand{\agentc}{\mathsf{agent\text{-}c}}
\newcommand{\agent}{\mathsf{agent}}
\newcommand{\gset}{G_{\text{DAG}}}
\newcommand{\error}{\mathsf{error}}
\newcommand{\inter}{\mathop{\mathrm{inter}}}
\newcommand{\iter}{\mathop{\mathrm{iter}}}
\newcommand{\diff}{\mathop{\mathrm{diff}}}
\newcommand{\supp}{\mathop{\mathrm{supp}}}
\newcommand{\DAGset}{\R_{\text{det}}}

\title{Robust forecast aggregation via additional queries}
\author{Rafael Frongillo}
\author[1]{Mary Monroe}
\author[2]{Eric Neyman}
\author[1]{Bo Waggoner}
\affil[1]{University of Colorado Boulder}
\affil[2]{Alignment Research Center}
\date{}

\begin{document}
\maketitle

\begin{abstract}
    We study the problem of robust forecast aggregation: combining expert forecasts with provable accuracy guarantees compared to the best possible aggregation of the underlying information.
%    We study the problem of robust forecast aggregation, where a principal seeks to aggregate expert forecasts to approximate an unknown quantity in the worst case over experts' underlying knowledge.
    Prior work shows strong impossibility results, e.g. that even under natural assumptions, no aggregation of the experts' individual forecasts can outperform simply following a random expert~\citep{neyman2022smarter}.
%    Prior work shows strong impossibility results: when experts report their posterior expectations, no aggregation rule can outperform selecting a random expert’s forecast in the worst case~\citep{neyman2022smarter}.

    In this paper, we introduce a more general framework that allows the principal to elicit richer information from experts through structured queries.
    Our framework ensures that experts will truthfully report their underlying beliefs, and also enables us to define notions of complexity over the difficulty of asking these queries. 
    Under a general model of independent but overlapping expert signals, we show that optimal aggregation is achievable in the worst case with each complexity measure bounded above by the number of agents $n$.
    We further establish tight tradeoffs between accuracy and query complexity: aggregation error decreases linearly with the number of queries, and vanishes when the ``order of reasoning'' and number of agents relevant to a query is $\omega(\sqrt{n})$.
    These results demonstrate that modest extensions to the space of expert queries dramatically strengthen the power of robust forecast aggregation.
    We therefore expect that our new query framework will open up a fruitful line of research in this area.
\end{abstract}

\section{Introduction} \label{sec:intro}

Aggregating distinct sources of information is a fundamental problem in statistics, machine learning, forecasting, and decision making.
We consider the problem of a principal aggregating forecasts from a set of experts into a single, provably-accurate forecast.
Traditional theoretical approaches often adopt specific parametric models of how agents acquire information, such as Gaussian signals with known covariance structures~\citep{ugander2015wisdom,lichtendahl2013wisdom}.
For example, works have concluded that common aggregation techniques like linear or logarithmic pooling are optimal for certain information structures; intuitively, linear pooling is better for aggregating redundant information, whereas logarithmic pooling is better for independent sources~\citep{laisney1985statistical,genest1984aggregating,goman2018efficient,lochert2010probabilistic,ranjan2010combining}.
Unfortunately, these works all depend on narrow assumptions about the information structure. These assumptions may fail in practice, or it may not be possible to know which of these assumptions holds.

%\raf{Consider adding an overview figure here}

In the theoretical computer science literature, recent works of \cite{arieli2018robust} and \cite{neyman2022smarter} offer an alternative: \emph{robust aggregation}.
They instead develop worst-case bounds on the performance of a given aggregation rule, under some weak conditions on the information structure, such as informational substitutes.
These results are exciting as they do not rely on narrow assumptions about how information is obtained.
Unfortunately, the results are also quite pessimistic: in the worst case, aggregating forecasts from $n$ agents is only a constant factor better than randomly picking one of their forecasts.
Moreover, these lower bounds are not driven by pathological examples, but appear to be an innate limitation of aggregation itself.
It may therefore seem impossible to design a general-purpose aggregation rule with favorable guarantees.

The key idea of this paper is a way around this impossibility: \emph{change the types of forecasts} we collect from agents.
Instead of simply asking agents to predict the (expected value of the) quantity of interest, what if we ask them to predict its higher moments, or even predict what other agents will say?
We illustrate with the following striking example.

\begin{example} \label{ex:common}
There are $n$ agents and $n+1$ signals $X_1,\dots,X_{n+1}$.
Each $X_j \sim \text{Normal}(0,1)$ independently, and the principal aims to estimate $Y = \sum_{j=1}^{n+1} X_j$.
Signal $X_{n+1}$ is a ``common signal'' observed by all agents.
Formally, each agent $i$ observes the pair $\cS_i = (X_i, X_{n+1})$, and submits their posterior expectation $Y_i = X_i + X_{n+1}$.
According to a lower bound by \cite{neymanthesis}, there is no aggregation function of the forecasts $(Y_1,\dots,Y_n)$ guaranteeing a nonvanishing improvement over the prior estimate of $0$.
%\cite{neymanthesis} shows that the lowest achievable error (which we define formally in \S~\ref{sec:prelim}) given these reports is lower bounded by $1 - \tfrac{4}{n}$. \bo{The problem here is the reader doesn't know if 1 is a good or bad error yet.}
In summary, this ``shared common signal'' setting is quite difficult under the standard approach.

But it turns out that by asking one additional question, we can optimally aggregate all information.
In addition to collecting $(Y_1,\dots,Y_n)$, we can simply ask any agent, say agent $1$, to predict the report of another agent, say agent $2$.
Then agent $1$'s answer is
\begin{align*}
    Q \coloneqq \E\left[Y_2 \mid S_1 \right]  
	=	\E\left[X_2 + X_{n+1} \mid X_1, X_{n+1} \right] 
	=	X_{n+1};
\end{align*}
that is, agent $1$'s answer to this question reveals the common signal.
So given one extra higher-order question, we can calculate all of the agent's private signals.
Namely, $X_{n+1} = Q$ and $X_i = Y_i - X_{n+1} = Y_i - Q$.
Then we can compute the optimal aggregation, $A^* = \sum_{j=1}^{n+1} X_j = Y$, with zero error.
\end{example}

The example above illustrates that, even with a single additional question, optimal aggregation is possible.
As we demonstrate in our results more generally, these additional forecasts, or ``queries'', can drastically improve the quality of the aggregation.
Motivated by this example, our broad goal is to study the following question.
\begin{quote}
  \textbf{Central question:}
%   \emph{What bounds can be guaranteed on aggregation performance given constraints on the complexity of the queries?}
    \emph{What is the optimal worst-case aggregation error given complexity constraints on the queries?}
\end{quote}

To define complexity constraints, we introduce a framework of \emph{DAG-elicitable query sets}.
Aside from constraining queries to ``reasonable'' ones, this framework has the nice benefit of guaranteeing that even strategic agents will answer their queries truthfully in equilibrium.
We introduce three notions of complexity for a DAG-elicitable query set to measure its usefulness in practice:
\begin{enumerate}\setlength{\itemsep}{0pt}
    \item \emph{query complexity}, the size of the query set;
    \item \emph{order complexity}, the ``longest chain of reasoning" an agent must follow in order to answer the query; and
    \item \emph{agent complexity}, the maximum number of agents whose beliefs are relevant to any query.
\end{enumerate}
Intuitively, while higher query complexity places more of a burden on the principal in collecting information from a large set of agents, higher order or agent complexity places a higher cognitive burden on experts in calculating their responses.

\subsection{Results}

In this paper, we study the central question above in the highly flexible yet simple ``partial information'' model due to \cite{satopaa2016modeling} where agents each observe a large subset of independent signals.
Specifically, each nonempty subset of agents $T \subseteq [n]$ observes a distinct random variable $X_T$, and the principal aims to estimate the variable $Y = \sum_T X_T$.

Our first result in \S~\ref{sec:optimal} is that one can achieve optimal error with DAG-elicitable queries.
Perhaps surprisingly, we find that optimal error is upper bounded by query, agent, and order complexity $n$; that is, we only need one query per expert to fully aggregate information.
A natural next question we address in \S~\ref{sec:approximate} is whether good performance is achievable under further restrictions on complexity measures.
First we study settings where query complexity is capped by $d \leq n$. 
We completely characterize achievable error in this class, where the error is a worst-case ratio that measures the amount of randomness the aggregation scheme fails to account for in $Y$.

\begin{theorem*}[Informal]
	If the principal is restricted to query complexity $d \leq n$, the best achievable error is $1 - \frac{d}{n}$. 
\end{theorem*}

If we instead restrict agent and order complexity by $d < n$ (e.g. in settings where agents are not able to conduct high-order reasoning about many others), the same linear upper bound is inherited. 
However, with additional queries, we can characterize even tighter bounds on error.
Here we restrict to a class of ``linear'' aggregation rules, which are natural for our setting given the principal's goal is to estimate the sum $\sum_T X_T$.

\begin{theorem*}[Informal]
    If the principal is restricted to agent and order complexity $d < n$, the best achievable error under linear aggregation rules is:
    \begin{enumerate}
        \item $1 - \Theta\left(\frac{d^2}{n}\right)$ if $d = o(\sqrt{n})$;
        \item $O\left(e^{\frac{-4d}{\sqrt{n}}}\right)$ if $d = \omega(\sqrt{n})$.
    \end{enumerate}
\end{theorem*} 
\noindent
That is, as long as agents are able to reason about $\omega(\sqrt{n})$ other agents and conduct $\omega(\sqrt{n})$-order reasoning when answering queries, high-quality aggregation is possible in the worst case. 

% \paragraph{Takeaways.}
Our complexity bounds paint a more optimistic picture of worst-case aggregation capabilities than previous work. 
Armed with a framework of additional queries, we find that low error rates can be achieved even when the principal is constrained to asking relatively simple queries. 
Our results can moreover be interpreted from the perspectives of both elicitation and information theory.
We define a reasonable query scheme that is both (1) incentive compatible, meaning the principal is able to capture the correct information from experts; and
(2) expressive enough for the principal to recover most of the experts’ information within a single aggregated value.
By introducing a more general model of the questions that the principal may ask experts, our paper opens the floor for exciting future research to better understand the power and limitations of robust aggregation.

\subsection{Related work}
\cite{arieli2018robust} first introduced the problem of robust forecast aggregation, where two conditionally independent (or Blackwell-ordered) experts submit their posterior beliefs about a binary state space. 
The authors characterize upper and lower bounds on an additive measure of error. 
\cite{levy2022combining} study a similar model when the principal has perfect knowledge of the marginal distribution over expert information, under restrictions on the correlation structure.
\cite{deoliviera21} also study robust aggregation when the output space of aggregation is finite.
We note that each of these papers consider \emph{probabilistic} forecasts over an event space, while we study aggregation of forecasts over real-valued quantities.

\cite{neyman2022smarter} study robust forecast aggregation with multiple experts, and give tight bounds on a multiplicative measure of performance when expert signals satisfy a condition called \emph{projective substitutes.}
Their model restricts forecasts to posterior expectations.
\cite{guo2025algorithmic} construct an FPTAS for optimal robust aggregation (as measured by additive error) when the class of information structures is finite.
When the information structure is continuous, the authors also derive an approximation algorithm for the two-agent, binary state setting of~\cite{arieli2018robust}, where agents report their Bayesian posteriors. 

The areas of peer prediction~\cite{faltingssurvey23} and elicitation complexity~\cite{casalaina-martin2017multi-observation,frongillo2021elicitation} also explore how to elicit information by asking agents additional types of questions.
Other aggregation works have proposed asking additional/different questions, such as multiple guesses~\cite{ugander2015wisdom}, an implicit game~\cite{lichtendahl2013wisdom}, or ``confidence''~\cite{frongillo2015elicitation}, but these works do not give robust aggregation guarantees and instead require narrow assumptions on the information structure.
The one exception to our knowledge is \cite{pan2024robust}, which applies peer prediction to robust aggregation, but in the simple setting of two experts and a binary outcome.

\section{Preliminaries} \label{sec:prelim}
\subsection{Problem setting} 
There are $n$ agents, called experts, with information relevant to a random variable $Y$ of interest.
Throughout, we let $[n] \coloneqq \{1,2,\dots,n\}$.
There is a set $\X \coloneqq (X_1, X_2, \ldots, X_m)$ of $m$ available signals, with the tuple $(Y, X_1, \ldots, X_m)$ jointly distributed from a prior distribution $\D$ that is known to all the experts. 
Throughout this paper, we do not assume that the principal has any knowledge of $\D$.
Each expert $i$ observes a subset $\cS_i \subseteq \X$.
For example, if each expert observes a distinct, single signal, $n = m$ and $\cS_i = (X_i)$; if all information is shared, $\cS_i = \X$. 
The posterior expectation of expert $i$, conditioned on their information, is $Y_i = \E [Y \mid \cS_i]$.

A principal aims to output an \emph{approximator} $A \coloneqq A(\D, X_1, \dots, X_m)$ of $Y$. 
Traditionally, the literature focuses on how a principal is able to aggregate the set of reports $(Y_1,\dots,Y_n)$.
For instance, $A = \frac{1}{n} \sum_{i=1}^n Y_i$ is an example of an approximator.
We consider a more general framework of queries.
Let $\Q = \bigcup_{i=1}^n \Q_i$ denote a finite list of queries, where each $Q \in \Q_i$ is a function $Q: (\cS_i, \D) \to \reals$ that maps expert $i$'s information to some real output.
\bo{Suggest making $\Q = (\Q_1,\dots,\Q_n)$ a finite list of queries}
We overload notation and let $Q = Q(\cS_i, \D)$ also denote the random variable output of a query when the inputs are clear.
We write $\agent(Q) = i$ to denote that the query $Q$ was answered by expert $i$. 
\bo{Would like to improve on the last two sentences, not sure how}

\begin{definition}
    A (deterministic) aggregation rule $R = (\Q, f)$ is a tuple with $\Q$ a list of queries and $f: \reals^{\Q} \to \reals$ a function from the query outputs to a real number.
    A randomized aggregation rule is a distribution $P$, where a random context $Z \sim P$ is drawn first from some sample space $\Omega$, determining the resulting aggregation rule $(\Q_Z, f_Z)$. 
    We denote $P_{\Q}$ as the marginal distribution over query sets, and $P_{\Q,f}$ as the marginal distribution over aggregation rules. 
    An approximator $A$ is \emph{implemented} by a randomized aggregation rule $P$ if $A = f_Z(\Q_Z)$ almost surely.
    \bo{is there anything wrong with just saying distribution over aggregation rules? I'll have to see once I see how we use this later, but it would be nice to simplify here if we can}
\end{definition}

\begin{example}[Random expert] \label{ex:random-experts}
    Let $\Q = \{Y_1, Y_2, \dots, Y_n\}$, i.e. the query set corresponds to each expert's expectation of $Y$; and let $P$ draw an expert uniformly at random, i.e. $P = U([n])$ where $U(S)$ denotes the uniform distribution over set $S$.
    Then for $i \sim P$, let $\Q_i = \{Y_i\}$, and $f_i(\Q_i) = \Q_i$.
    That is, $P$ is a randomized aggregation rule which outputs a uniformly random expert's expectation of $Y$.
    We call $P$ the ``random expert'' aggregation rule.
\end{example}

Note that we only allow ex-ante randomization over aggregation rules, so that e.g. the principal cannot condition on the realization of a subset of queries. 

\paragraph{Error.}
Given an approximator $A$, let $L(A) \coloneqq \E \left[ (A - Y)^2 \right]$ be the expected squared error of $A$.
The optimal aggregation of all information is denoted $A^* = \E[Y \mid X_1,\dots,X_m]$.
Because $A^*$ is the Bayesian posterior expectation conditioned on all available information, the lowest possible expected error is $L(A^*)$.%
\footnote{We assume without loss of generality that each $X_j$ is observed by at least one agent, so that the agents' pooled information determines $A^*$. Otherwise, $L(A^*)$ is not achievable.}
A baseline to compare algorithms against is $L(\mu)$ where $\mu = \E[Y]$ is the prior expectation of $Y$; that is, the variance of $Y$.
Following~\cite{neyman2022smarter}, we define the worst-case error
\mary{insert footnote about additive vs. multiplicative measure of performance}
of an approximator $A$ to be
\begin{equation} \label{eqn:performance}
  \error(A) \coloneqq \max_{\D} \frac{L(A) - L(A^*)}{L(\mu) - L(A^*)} .
\end{equation}
That is, the error expresses the worst-case difference in information expressed by $A$ versus $A^*$, normalized by the relative information in $A^*$.
We have $\error(\mu) = 1$ and $\error(A^*) = 0$.

\begin{examplecon}{\ref{ex:random-experts}}
    Consider the approximator $A_P$ implemented by the random experts aggregation rule $P$.
    Then $L(A_P) = \E_{\D, i \sim P} [(f_i(\Q_i) - Y)^2] = \E_{\D, i \sim U([n])} [(Y_i - Y)^2]$.
    \bo{would be nice to have the example also illustrate $\error$, but then we need a whole information structure.}
\end{examplecon}

\subsection{DAG-elicitable queries}
As defined, an aggregation rule allows the principal to ask any set of queries.
However, in reality experts may not be incentivized to answer these queries truthfully according to their private information.
We are thus motivated to restrict admissable query sets in a way that guarantees agents will answer queries truthfully and accurately.
Moreover, we aim to define a framework that allows us to articulate the \emph{complexity} of an approximator, given there may be many different aggregation rules that implement it.

Following the elicitation literature, we define a query function as elicitable if there is a payment scheme that incentivizes the expert to report the output truthfully.
\begin{definition}
    A query $Q$ is \emph{elicitable} with respect to a set $\R$ of real-valued random variables if there exists a payment function $\pi(r, \R)$ such that $Q(\cS_i, \D) = \arg\max_{r \in \reals} \E_{\D} [\pi(r, \R) \mid \cS_i]$.
\end{definition}

For example, it is well-known that the quadratic score $\pi(r, Y) \coloneqq 1 - (r - Y)^2$ elicits the mean, i.e. $\E[Y \mid \cS_i] = \arg\max_r \E_{\D} [1 - (r - Y)^2 \mid \cS_i]$, so that $Y_i$ is elicitable with respect to $Y$.
But one could take this scheme a step further and elicit a query such as $\E[Y_i \mid \cS_j]$, or expert $j$'s expectation of $i$'s expectation of $Y$, with the payment $\pi(r, Y_i) \coloneqq 1 - (r - Y_i)^2$. 
Then, since expert $i$ is strictly incentivized to report truthfully, so is expert $j$. 
We can thus generalize elicitability of a query set beyond simply scoring reports against $Y$ as follows. 

\begin{definition}{(DAG-elicitable)}
    \maryt{write as ``with respect to a class of queries Q,'' then only focus on expectation query class}
    A set of queries $\Q$ is \emph{DAG-elicitable} if we can arrange the queries into a directed acyclic graph $G = (V, E)$ with $V \coloneqq \Q \cup \{ Y \}$ such that 
    \begin{enumerate}
        \item each $Q \in \Q$ is elicitable with respect to the set $N_G(Q)$, where $N_G(Q)$ is the set of out-neighbors of $Q$ (which may include $Y$), and
        \item $Y$ is the unique sink. 
    \end{enumerate}
    If such a graph $G$ exists, we say it DAG-elicits the query set $\Q$.
\end{definition}

We will refer to an aggregation rule $R = (\Q,f)$ as DAG-elicitable if $\Q$ is DAG-elicitable, and a randomized aggregation rule $P$ as DAG-elicitable if each $\Q \in \supp(P_{\Q})$ is DAG-elicitable for $P_{\Q}$ the marginal over query sets. 
While the structure of a DAG-elicitable aggregation rule implies each expert's payment depends on the reports of others, we can treat the rule as a one-shot mechanism and calculate all payments at once. 
The following is then a straightforward observation. 

\begin{observation}
    Consider the Bayesian game induced by a DAG-elicitable aggregation rule and any graph $G$ that DAG-elicits it: each agent has private information $\cS_i$, submits a report $r_i: (\cS_i,\D) \to \reals$, and receives payment $\pi(r_i, N_G(Q))$.
    Then it is a unique Bayes-Nash equilibrium for experts to report answers to each query truthfully, i.e. $r_i = Q(\cS_i, \D)$. 
\end{observation}

\paragraph{Complexity measures.}
Consider a DAG-elicitable query set $\Q$. 
One can then ask how efficient it is to elicit each query, for both the principal and the experts. 
Motivated by this question, we define several notions of complexity over $\Q$ relative to the graphs $G$ that DAG-elicit the query set. 
Specifically, we identify three independent axes along which to measure complexity.
\begin{enumerate}
    \item \emph{Query size}: The total number of questions asked, $|\Q|$.
	\item \emph{Order complexity}: the ``longest chain of reasoning'' involved in any query, i.e., the highest order of any query.
	A query is order $1$ if it is scored against $Y$, is order $2$ if it is scored against order $1$ queries and $Y$, and so on (see Figure~\ref{fig:complexity}).
    Intuitively, the order complexity of an aggregation rule is the maximum number of iterated beliefs involved.
    A higher order complexity implies a larger cognitive burden on experts in answering their queries. 
	\item \emph{Agent complexity}: The maximum number of agents whose beliefs are relevant to any query.
	Suppose we can write a query as $Q = Q(S_i, \D')$ where $\D'$ is the marginal distribution on $Y$ and the set of signals observed by some group of $k$ agents.
    If we can do so for $k$ agents but not $k-1$, then the agent complexity of $Q$ is $k$.
	The agent complexity of $\Q$ is the maximum over any of the queries.
	Similar to order complexity, large agent complexity can be impractical, as it requires agents to form precise beliefs about the beliefs and behavior of many other agents.
\end{enumerate}

Formally, let $\gset(\Q)$ be the set of graphs that DAG-elicit a set of queries $\Q$.
For each $Q, Q' \in \Q$, let $Q \to_G Q' \coloneqq \ones \{\text{$Q'$ reachable from $Q$ in $G$}\}$, and let $N^*_G(Q) \coloneqq \{Q' \in \Q: Q \to_G Q'\}$ denote the set of reachable nodes from $Q$ in $G$, including $Q$ itself.
Then let $I_G^*(Q) \coloneqq \{\agent(Q'): Q' \in N_G^*(Q)\}$ represent the set of experts whose beliefs are relevant to $Q$; we call $I_G^*(Q)$ the \emph{agent complexity} of query $Q$ under graph $G$.
Let $\ell_G(Q)$ denote the \emph{longest} path distance from query $Q$ to $Y$ in graph $G$. 
We will sometimes refer to $\ell_G(Q)$ as the \emph{order} of query $Q$ under graph $G$.
Using this notation, we can formally define the aforementioned complexity measures of DAG-elicitable query sets.

\begin{definition}{(Complexity notions)} 
    Take some DAG-elicitable query set $\Q$.
    We define the query complexity $\queryc(\Q)$, order complexity $\orderc(\Q)$, and agent complexity $\agentc(\Q)$ as follows:
    \[
    \queryc(\Q) \coloneqq |\Q|, \quad \orderc(\Q) \coloneqq \min_{G \in \gset(\Q)} \max_{Q \in \Q} \ell_G(Q), \quad \agentc(\Q) \coloneqq \min_{G \in \gset(\Q)} \max_{Q \in \Q} |I_G^*(Q)|.
    \]
    In an an overload of notation, for a random aggregation rule $P$, we let each of its complexities be the the maximum over all query sets in $\Q \in \supp(P_{\Q})$. That is, 
$\queryc(P) \coloneqq \max_{\Q \in \supp(P_{\Q})} \queryc(\Q)$, $\orderc(P) \coloneqq \max_{\Q \in \supp(P_{\Q})} \orderc(\Q)$, and $\agentc(P) \coloneqq \max_{\Q \in \supp(P_{\Q})} \agentc(\Q)$.
\end{definition}

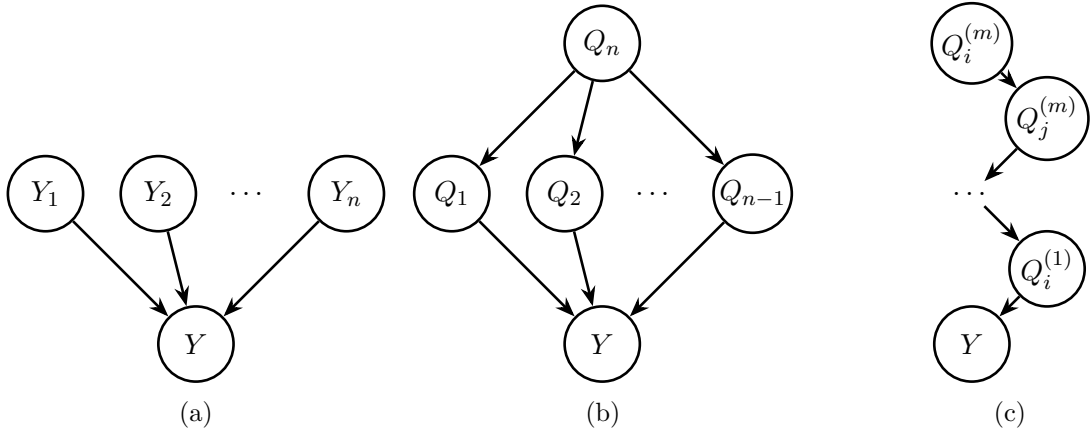
\begin{figure}[t!]
\centering
\begin{subfigure}{0.32\linewidth}
    \centering
    \begin{tikzpicture}[
        node distance=1.5cm,
        querynode/.style={draw, circle, minimum size=10mm,inner sep=1pt},
        line width=1pt,
        >=Stealth
    ]
    % Nodes
    \node[querynode] (Y1) at (-2,0) {$Y_1$};
    \node[querynode] (Y2) at (-0.5,0) {$Y_2$};
    \node (D) at (0.7, 0) {$\dots$};
    \node[querynode] (Y3) at (2,0) {$Y_n$};
    \node[querynode] (Y)  at (0,-2) {$Y$};

    % Edges to Y
    \draw[->] (Y1) -- (Y);
    \draw[->] (Y2) -- (Y);
    \draw[->] (Y3) -- (Y);
    \end{tikzpicture}
\caption{}
\label{subfig:standard-complexity}
\end{subfigure}
\begin{subfigure}{0.32\linewidth}
    \centering
    \begin{tikzpicture}[
        node distance=1.5cm,
        querynode/.style={draw, circle, minimum size=10mm,inner sep=1pt},
        line width=1pt,
        >=Stealth
    ]
    % Nodes
    \node[querynode] (Q4) at (0,2) {$Q_n$};
    \node[querynode] (Q1) at (-2,0) {$Q_1$};
    \node[querynode] (Q2) at (-0.5,0) {$Q_2$};
    \node (D) at (0.7, 0) {$\dots$};
    \node[querynode] (Q3) at (2,0) {$Q_{n-1}$};
    \node[querynode] (Y)  at (0,-2) {$Y$};

    % Edges from Q4
    \draw[->] (Q4) -- (Q1);
    \draw[->] (Q4) -- (Q2);
    \draw[->] (Q4) -- (Q3);

    % Edges to Y
    \draw[->] (Q1) -- (Y);
    \draw[->] (Q2) -- (Y);
    \draw[->] (Q3) -- (Y);

    \end{tikzpicture}
\caption{}
\label{subfig:high-agent-complexity}
\end{subfigure}
\begin{subfigure}{0.32\linewidth}
    \centering
    \begin{tikzpicture}[
        node distance=1.5cm,
        querynode/.style={draw, circle, minimum size=10mm, inner sep=1pt},
        line width=1pt,
        >=Stealth
    ]
    % Nodes
    \node[querynode] (Q4) at (0,2) {$Q_i^{(m)}$};
    \node (Q3)[querynode] at (1, 1) {$Q_j^{(m)}$};
    \node (Q2) at (0,0) {$\dots$};
    \node[querynode] (Q1) at (1,-1) {$Q_i^{(1)}$};
    \node[querynode] (Y)  at (0,-2) {$Y$};

    \draw[->] (Q4) -- (Q3);
    \draw[->] (Q3) -- (Q2);
    \draw[->] (Q2) -- (Q1);
    \draw[->] (Q1) -- (Y);

    \end{tikzpicture}
\caption{}
\label{subfig:high-order-complexity}
\end{subfigure}
\caption{Three graphs that DAG-elicit different query sets and upper bound their complexity measures; each query $Q$ is indexed in subscript by $\agent(Q)$. 
Figure~\ref{subfig:standard-complexity} depicts a graph implementing the standard query set $\Q = (Y_1,\dots,Y_n)$; $\Q$ has query size $n$, elicitation order 1, and agent complexity 1.
Figure~\ref{subfig:high-agent-complexity} depicts an example of high agent complexity ($n$) and low order complexity ($2$), with the query $Q_n = \E[\sum_{j =1}^{n-1} Y_j \mid \cS_n]$.
In other words, we ask expert $n$ to predict the sum of all other experts' posterior expectations.
Figure~\ref{subfig:high-order-complexity} depicts an example with high order complexity ($2m$) but low agent complexity ($2$).
$Q_i^{(m)}$ represents following iterated-expectation query: agent $i$'s expectation of agent $j$'s expectation of agent $i$'s expectation of \dots of agent $j$'s expectation of $Y$.
}
\label{fig:complexity}
\end{figure}

\subsection{Partial information model}
The \emph{partial information model} is the special case of the general information model, where the signals are independent mean-zero variables and $Y$ is their sum.
Without loss of generality, we suppose there are $2^n$ independent signals, each one corresponding to a particular subset of the experts. 
Specifically, let $X_T \in \reals$ be an independent random variable drawn from a distribution $\D_T$ for each subset $T \subseteq [n]$, where $\mu_T \coloneqq \E_{\D}[X_T] = 0$ and $\Var_{\D}(X_T) \coloneqq \E_{\D}[X_T^2] < \infty$.
For mathematical convenience we define $X_{\emptyset} = 0$ with probability one.
We assume there exists some $T \subseteq [n]$ such that $\Var(X_T) > 0$; otherwise the aggregation problem is trivial. 
The prior distribution $\D = ( \D_T : T \subseteq [n])$ is known to the experts. 

Each random variable $X_T$ is revealed only to the set of experts $T$. 
That is, each expert $i$ sees signal $\cS_i = \{ X_T : i \in T \}$. 
Here, $X_{\{i\}}$ represents information private only to expert $i$, and $X_{[n]}$ represents information that is shared across all experts. 
Then the aggregator aims to estimate $Y = \sum_{T \subseteq [n]} X_T$. 
(Example~\ref{ex:common} is thus a special case of the partial information model.) 

Note a forecaster $i$'s posterior expectation of $Y$, $Y_i \coloneqq \E[ Y \mid \cS_i]$,  is simply the sum of their observed signals:
\[ \E[Y \mid \cS_i] = \E [ \sum_{T \subseteq [n]} X_T \mid \{X_T: i \in T\}] = \sum_{T \subseteq [n]: i \in T} X_T.\]

Our results in this paper focus solely on the partial information model, which captures a multitude of settings where experts have overlapping information about the variable of interest. 
We leave extensions to correlated signal structures as an exciting avenue for future work. 
\mary{insert discussion of nonzero mean variables here}

\paragraph{Error.}
In the partial information model, the optimal approximator $A^*$ satisfies $A^* = \sum_{T \subseteq [n]} X_T$, so that $L(A^*) = 0$. 
It follows that we can write the error of an approximator $A$ under the partial information model as
\begin{equation} \label{eqn:partial-performance}
  \error(A) = \max_{\D} \frac{L(A)}{\Var(Y)} = \max_{\D} \frac{\E\left[ (A - \sum_T X_T)^2 \right]}{\sum_T \Var(X_T)}. 
\end{equation}
That is, $\error(A)$ measures the amount of randomness that $A$ fails to account for in $Y$ (as measured by variance), in the worst case over the distribution $\D$.

\section{Optimal Aggregation} \label{sec:optimal}

We begin by asking if optimal aggregation is even possible in the DAG-elicitable framework, and if so, whether there are reasonable upper bounds on the complexity of aggregation rules that implement $A^*$.
It will be useful in this quest to start by defining a specific class of aggregation rules. 

\subsection{Linear aggregation rules}
Given the goal of the principal is to estimate $Y = \sum_T X_T$, a natural class of approximators to consider are of the linear form $\sum_T \beta_T X_T$. 
It turns out one can implement such functions using an DAG-elicitable class of rules we call the \emph{linear aggregation rules}.

\begin{definition}
    A \emph{linear aggregation rule} $R = (\Q, f)$ is an DAG-elicitable aggregation rule where, for some graph $G$ that DAG-elicits $\Q$, each query $Q \in \Q$ can be written as 
    \[Q = \E [ \textstyle\sum_{Q' \in N_G(Q)} \alpha_{Q,Q'} Q' \mid \cS_i ] \]
    for some set of constants $\{\alpha_{Q,Q'}\}_{Q' \in N_G(Q)}$;
    and $f(\Q) = \sum_{Q \in \Q} \beta_Q Q$ for some set of real-valued constants $\{\beta_Q\}_{Q \in \Q}$.

    We define $P$ as a \emph{randomized linear aggregation rule} if $(\Q, f)$ is a linear aggregation rule for each $(\Q, f) \in \supp(P_{Q,f})$.
\end{definition}

\mary{note here that adding a constant $c$ will only decrease performance for a linear aggregation function.}

We note that implicit in the definition, each query $Q \in \Q$ is elicitable using the payment rule $\pi(Q, N_G(Q)) = 1 - (Q - \sum_{Q' \in N_G(Q)} \alpha_{Q,Q'} Q')^2$.
Moreover, using induction, we can show that any linear aggregation rule outputs a weighted linear combination of the $X_T$'s.

\begin{proposition} \label{prop:linear-approximators}
    Any linear aggregation rule outputs an approximator $A = \sum_T c_T X_T$ for some set of constants $\{c_T\}_{T \subseteq [n]}$. 
\end{proposition}

\begin{proof}
We start by proving by induction that any linear query $Q$ of order $d$ in a linear aggregation rule can be expressed as $\sum_T c_{Q,T} X_T$ for some constants $\{ c_{Q,T} \}_{T \subseteq [n]}$.
Let $d = 1$. 
Then a query of order $d$ can only be scored against the sink, $Y$. 
That is, if $\agent(Q) = i$, $Q = \E\left[ \alpha_{Q, Y} Y \mid \cS_i\right] = \sum_{T \ni i} \alpha_{Q,Y} X_T = \sum_{T} c_{Q,T} X_T$ for $c_{Q,T} \coloneqq \ones_{i \in T} \alpha_{Q,Y}$.
Now, for $1 \leq l \leq d$, assume any query $Q$ of order $d$ in the linear aggregation rule can be expressed as $Q = \sum_T c_{Q,T} X_T$. 
Take any query $Q \in \Q$ of order $d+1$ with $\agent(Q) = i$, and consider a DAG $G$ that DAG-elicits the rule.
Then 
\begin{align*}
    Q = \E\left[ \textstyle \sum_{Q' \in N_G(Q)} \alpha_{Q,Q'} Q' \mid \cS_i\right] &= \E\left[ \textstyle \sum_{Q' \in N_G(Q)} \alpha_{Q,Q'} \left( \sum_T c_{Q',T} X_T \right) \mid \cS_i\right] \\
    &= \textstyle \sum_{Q' \in N_G(Q)} \alpha_{Q,Q'} \left( \sum_{T \ni i} c_{Q',T} X_T \right) \\
    &= \textstyle \sum_T c_{Q,T} X_T,
\end{align*} 
where the first equality follows by the inductive hypothesis; and $c_{Q,T} \coloneqq \ones_{i \in T}  \sum_{Q' \in N_G(Q)} \alpha_{Q, Q'} c_{Q',T}$.

It follows that each query can be expressed as $Q = \sum_T c_{Q,T} X_T$, and therefore $f(\Q) = \sum_Q \beta_Q Q = \sum_Q \beta_Q \sum_T c_{Q,T} X_T = \sum_T c_T X_T$ for $c_T \coloneqq \sum_Q \beta_Q c_{Q,T}$.
\end{proof}

\begin{example} \label{ex:linear}

    Consider the setting with three experts and three queries, indexed by each expert: $Q_1(\cS_1, \D) = \E[ Y \mid \cS_1]$, $Q_2(\cS_2, \D) = \E[ Y \mid \cS_2]$, and $Q_3(\cS_3, \D) = \E[ \alpha_1 Q_1 + \alpha_2 Q_2 \mid \cS_3]$ for $\alpha_1 = -1, \alpha_2 = 2$; and $f(\Q) = \sum_{Q \in \Q} Q$.
    Then as shown in Appendix~\ref{app:optimal}, $ f(\Q) = \sum_{T \subseteq [n]} \beta_T X_T$ with
    \begin{align*}
    \beta_T =
    \begin{cases}
        0, & T \cap \{1,2,3\} = \emptyset, \{3\}, \{1,3\} \\
        1, & T \cap \{1,2,3\} = \{1\}, \{2\} \\
        2, & T \cap \{1,2,3\} = \{1, 2\} \\
        3, & T \cap \{1,2,3\} = \{2, 3\}, \{1,2,3\}
    \end{cases}~.
    \end{align*}
\end{example}

\subsection{Intersection queries}

We now consider a specific class of linear aggregation rules which will be useful for achieving low error. 
Generalizing the first-order queries $\E[ Y \mid \cS_i]$, we can consider queries asked to expert $i$ about their expectation of another expert $j$'s report.
Formally, we can ask for
\begin{align*}
    \E[Y_j \mid \cS_i]
	= \E[\textstyle\sum_{T\subseteq[n] : j \in T} X_T \mid \{X_T : i \in T\}] 
	= \displaystyle\sum_{T\subseteq[n]: i,j \in T} X_T .
\end{align*}
Intuitively, $i$ knows that $j$'s expectation is the sum of the signals that $j$ observes.
Of those, $i$ also observes some, but the rest are independent and mean-zero and drop out.

We can then generalize this approach to higher orders. 
Inductively, we define \emph{iterated expectation queries} as follows: $\iter(i) \coloneqq Y_i$ and $\iter(i_1,\dots,i_k) \coloneqq \E\left[ \iter(i_1,\dots,i_{k-1}) \mid \cS_{i_k} \right]$ for ordered list $(i_1,\dots,i_k)$.
In other words, $\iter(i_1,\dots,i_k)$ is agent $i_k$'s expectation of agent $i_{k-1}$'s expectation of \dots of agent $i_1$'s posterior expectation of $Y$.
We can immediately observe that in the partial information model, only the unordered set of indices $S = \{i_1,\dots,i_k\}$ matters.

\begin{proposition} \label{prop:inter-queries}
    Let $L \coloneqq (i_1,\dots,i_k)$.
    Then
    $\iter(L) = \sum_{T: S \subseteq T} X_T$ for $S = \{i_1, \dots, i_k\}$.
\end{proposition}

\begin{proof}
    We proceed by induction over the length $k$ of $L$.
    Let $k = 1$.
    Then, by definition, $\iter(i)  = \sum_{T\subseteq[n]: \{i\} \subseteq T} X_T$.
    Now assume the statement holds for all ordered lists $L$ with $1 \leq |L| \leq k$. 
    Take an ordered list $L' = (i_1, \dots,i_{k+1})$. 
    Then $\iter(i_1,\dots,i_{k+1}) = \E\left[ \iter(i_1,\dots,i_k) \mid \cS_{i_{k+1}} \right] = \E\left[ \sum_{T: \{i_1, \dots, i_k\} \subseteq T} X_T \mid \cS_{i_{k+1}} \right] = \sum_{T: \{i_1, i_2, \dots, i_{k+1}\} \subseteq T} X_T$.

\end{proof}

In the context of the partial information model, then, we can define \emph{intersection queries} as $\inter(\{i\}) \coloneqq Y_i$ and $\inter(S) \coloneqq \sum_{T: S \subseteq T} X_T$.
We define an \emph{intersection query set} of order $d$ as $\Q_{\inter}(d) \coloneqq \{ \inter(S): |S| \leq d\}$.

\begin{proposition} \label{prop:inter-query-order}
    $\Q_{\inter}(d)$ is DAG-elicitable, with $\queryc(\Q_{\inter}(d)) = 2^d$, $\orderc(\Q_{\inter}(d)) \leq d$, and $\agentc(\Q_{\inter}(d)) \leq d$.
\end{proposition}

\begin{proof}
    By construction, there is a graph $G$ that DAG-elicits the query set, where each query $\inter(S)$ for $|S| = k$ shares an edge with a query $\inter(S')$ for $|S| = k - 1$. 
    Formally, if $k = 1$ and $\agent(Q) = i$, then $\inter(S) = Y_i$, so that we can elicit $Q$ with payment rule $\pi(\inter(\{i\}), \Q \cup \{Y\}) = 1 - (Q - Y)^2$.
    Then we add an edge $(Q, Y)$ in $G$.
    Now, for $k > 1$, consider query $Q = \inter(S)$ with $|S| = k$ and $\agent(Q) = i$.
    By definition, $\inter(S) = \E\left[ \inter(S \backslash \{ i \}) \mid \cS_i \right]$, so we can elicit $Q$ with payment rule
    $\pi(\inter(S), \Q \cup \{Y\}) = 1 - (Q - \inter(S \backslash \{ i \}))^2$. 
    Then we add an edge $(Q, Q')$ in $G$, where $Q' = \inter(S \backslash \{ i \})$.
    It follows by induction that a query $Q = \inter(S)$ with $|S| = k$ has order and agent complexity $k$ in $G$, so that $\max_{Q \in \Q} \ell_G(Q) = d$ and $\max_{Q \in \Q} |I_G^*(Q)| = d$.
    Moreover, $\queryc(\Q_{\inter}(d)) = |\queryc(\Q_{\inter}(d))| = 2^d$. 
\end{proof}

We define an \emph{intersection aggregation rule} as an aggregation rule with an intersection query set and a linear aggregation function.

\begin{definition}
    A (deterministic) intersection aggregation rule $R = (\Q, f)$ is an DAG-elicitable rule with $\Q \coloneqq \Q_{\inter}(d)$, and $f(\Q) = \sum_{Q \in \Q} \beta_Q Q$ for some set of real-valued constants $\{\beta_Q\}_{Q \in \Q}$.
    We define $P$ as a \emph{randomized linear aggregation rule} if $(\Q, f)$ is an intersection aggregation rule for each $(\Q, f) \in \supp(P_{Q,f})$.
\end{definition}

\subsection{Optimal aggregation with intersection queries}
Armed with intersection queries, we can now show that optimal aggregation is possible under the DAG-elicitable framework. 
We leave the proof to Appendix~\ref{app:optimal}: at a high level, we are able to completely reconstruct each random variable $X_T$ using the inclusion-exclusion principle by summing up intersection queries for each $S \supseteq T$.
It follows that the query complexity of this aggregation rule is $2^n$.

\begin{theorem} \label{thm:opt-intersection}
    There exists a deterministic intersection aggregation rule $R$ that implements $A^*$ with $\agentc(R) \leq n$, $\orderc(R) \leq n$, and $\queryc(R) = 2^n$.
\end{theorem}

\subsection{Difference queries}
We define another class of linear aggregation rules called \emph{difference aggregation rules}.
Again, we consider some ordered list $L = (i_1, i_2, \ldots, i_k)$ of agents. 
Then we inductively define \emph{difference queries} as follows: $\diff(i) \coloneqq Y_i$ and $\diff(i_1,\dots,i_k) \coloneqq\E[Y - \sum_{1 \leq j < k} \diff(i_1, \dots, i_{k-j}) \mid \cS_k]$.
In the partial information model, we can show that $\diff(i_1,\dots,i_k)$ exactly expresses the expected value of $Y$ conditioned on the \emph{set difference} between $\cS_{i_k}$ and $\{\cS_{i_\ell}\}_{\ell=1}^{k-1}$:

\begin{proposition} \label{prop:diff-queries}
    Let $L \coloneqq (i_1,\dots,i_k)$.
    Then $\diff(L) = \E[Y \mid \cS_{i_k} \backslash \cup_{\ell < k} \cS_{i_{\ell}} ]$. 
\end{proposition}

\begin{proof}
    We prove by induction over $k$ that for $|L| = k$,  $\diff(L) =  \E[Y \mid \cS_{i_k} \backslash \cup_{\ell < k} \cS_{i_{\ell}} ]$. 
    Let $k = 1$. 
    Then $\diff(L) = \E[Y \mid \cS_{i_1}] = \E[Y \mid \cS_{i_1} \backslash \emptyset]$.
    Now, assume the statement holds for all ordered lists $L$ with $1 \leq |L| \leq k$.
    Then consider an ordered list $L' = (i_1, \dots, i_{k+1})$ with $|L| = k + 1$. 
    We have
    \begin{align*}
        \diff(L') &= \E [Y \mid \cS_{i_{k+1}}] - \E[ \textstyle\sum_{1 \leq j \leq k} \diff(i_1, \dots, i_k) \mid  \cS_{i_{k+1}}] \\
        &= \E\left[Y \mid \cS_{i_{k+1}}\right] - \sum_{1 \leq j \leq k} \E\left[ \E \left[Y \mid \cS_{i_j} \backslash \cup_{\ell < j} S_{\ell_j}\right] \mid \cS_{i_{k+1}}\right] \\
        &= \sum_{T \ni i_{k+1}} X_T - \sum_{1 \leq j \leq k} \quad \sum_{\substack{T:  \{i_{k+1}, i_j\} \subseteq T, \\ i_{\ell} \notin T \; \forall \ell < j}} X_T \\
        &= \sum_{\substack{T: i_{k+1} \in T, \\i_{\ell} \notin T \; \forall \ell < k+1}} X_T \\
        &= \E[Y \mid \cS_{i_{k+1}} \backslash \cup_{\ell < k+1} S_{i_{\ell}}].
    \end{align*}
    The fourth equality holds because each subset $T$ satisfying $i_{\ell}, i_{k+1} \in T$ for some $\ell < k+1$ must be covered exactly once in the right-hand sum. 
    That is, $T$ it uniquely identified by the earliest index $\ell$ such that $i_{\ell} \in T$. 
\end{proof}

We define a \emph{difference query set} over an ordered list $L \coloneqq (i_1, \dots, i_d)$ of agents as $\Q_{\diff}(L) \coloneqq \{ \diff(i_1), \diff(i_1, i_2), \dots, \diff(i_1, \dots, i_d) \}.$
We can then show that $\Q_{\diff}(L)$ is DAG-elicited by an outdegree-1 DAG $G$, so that each of its complexity measures is bounded by $|L|$ (see Appendix~\ref{app:optimal} for formal details).

\begin{proposition} \label{prop:diff-query-order}
    $\Q_{\diff}(L)$ is DAG-elicitable, with $\orderc(\Q_{\diff}(L)) = |L|$, $\orderc(\Q_{\diff}(L)) \leq |L|$, $\agentc(\Q_{\diff}(L)) \leq |L|$.
\end{proposition}

\begin{definition}
    A (determistic) difference aggregation rule $R = (\Q, f)$ is an DAG-elicitable rule with $\Q = \Q_{\diff}(L)$ for some ordered list of agents $L$, and $f(\Q) = \sum_{Q \in \Q} \beta_Q Q$ for some set of real-valued constants $\{\beta_Q\}_{Q \in \Q}$.
    We define $P$ as a \emph{randomized difference aggregation rule} if $(\Q, f)$ is a difference aggregation rule for each $(\Q, f) \in \supp(P_{Q,f})$.
\end{definition}

\subsection{Optimal aggregation with difference queries}
In Theorem~\ref{thm:opt-intersection}, we showed that optimal aggregation is achievable with query complexity $2^n$. 
While it seems plausible that the principal would need all $2^n$ queries to fully aggregate $Y$, given there are $2^n$ random variables embedded in its sum, we find that one can actually decrease query complexity of optimal aggregation significantly. 
Intuitively, difference queries allow us to ask each expert what information they would \emph{add} to the current estimation. 
By induction over difference queries up to expert $n$, this approach allows us to fully re-construct $Y$ with only $n$ queries. 

\begin{theorem} \label{thm:opt-difference}
    There exists a deterministic difference aggregation rule $R$ that implements $A^*$ with $\agentc(R) \leq n$, $\orderc(R) \leq n$, and $\queryc(R) = n$.
\end{theorem}

\begin{proof}
    Consider the DAG-elicitable query set $\Q_{\diff}(L)$ for $L = (1,2, \dots, n)$, which by Proposition~\ref{prop:diff-query-order} is of agent and order complexity at most $n$, and query complexity $n$.
    Let $Q_i \coloneqq \diff(1,2,\dots,i)$, so that $\Q_{\diff}(L) = (Q_1, Q_2,\dots, Q_n)$.
    By Proposition~\ref{prop:diff-queries}, 
    \[Q_i = \E[ Y \mid \cS_i \backslash \cup_{j < i} \cS_j] =  \sum_{\substack{T: i \in T, \\j \notin T \; \forall j < i}} X_T.\]
    Since each $Q_i$ collects exactly those $X_T$ whose smallest index in $T$ is $i$, and every subset $T \subseteq [n]$ corresponds to a some unique smallest element $i$, the family of queries $\{Q_i\}_{i=1}^n$ partition and cover the $X_T$'s, and we have
    $A = \sum_{i\in[n]} Q_i = Y$.
    It follows immediately that $\error(A) = 0$ with probability one.

    % moreover, $Q_n = \E[ Y \mid \cS_n \backslash \cup_{j < n} S_j]$ must have agent and order complexity $n$ in any graph $G$ that DAG-implements $\Q$. \mary{TODO: prove.}
\end{proof}

\section{Approximate aggregation} \label{sec:approximate}

Since we have shown that optimal aggregation is possible with agent, order, and query complexity $n$, a logical next step is to ask if good error rates are still achievable with lower complexity values.
In this section, we restrict approximators to those achievable by aggregation rules of (1) limited query complexity or (2) limited agent and order complexity.
We then give upper and lower bounds on error in both regimes. 

\subsection{Bounded query complexity} \label{subsec:query}
One can imagine a principal is constrained to only sample a small number $d < n$ of queries from the experts when attempting to pool information.

Formally, let $\Aquery$ denote the set of approximators $A$ that are implented by DAG-elicitable, randomized aggregation rules restricted to query complexity at most $d$.
While it seems plausible that the principal has no hope in achieving good approximations under limited query complexity, we find in Theorem~\ref{thm:limited-query-error} a clean linear relationship between $d$ and achievable error.

\begin{theorem} \label{thm:limited-query-error}
    Let $d < n$. Then
    $\min_{A \in \Aquery} \error(A) = 1 - \frac{d}{n}$.
\end{theorem}

That is, each additional query reduces the best achievable error by $1/n$, and thus carries the same informational value. 
This result is therefore optimistic, in that the principal is guaranteed steady, predictable gains over additional queries. 
On the other hand, if the principal is constrained to a small number of queries, there are no early exponential improvements in adding one or two more queries to the set.

We leave the proof to Appendix~\ref{app:query} and give a sketch here. 
To prove Theorem~\ref{thm:limited-query-error}, we first show that $1 - \frac{d}{n}$ is a lower bound on error for approximators in $\Aquery$. 
Specifically, consider an adversarial distribution $\D$ which puts all variance on the singleton variables $\{X_{\{i\}}\}_{i=1}^n$.
Applying the minimax inequality, we allow the principal to see $\D$ first and use a deterministic aggregation rule to implement an approximator. 
Since these aggregation rules must be of query complexity at most $d$, they must also be of agent complexity at most $d$.
Even if the principal can see $\D$ first, the best they can do is take the expected sum of singleton set variables $\{X_{\{i\}}\}_{i \in S}$ for some set $S$ of agents with size at most $d$.
It follows that the variance of the $n-d$ remaining singleton variables under $\D$ induce an error of $1-\frac{d}{n}$.

We then show there exists an approximator $A \in \Aquery$ that achieves the lower bound.
This approximator is implemented by a randomized difference aggregation rule: specifically, we sample a subset $S$ of agents with size $d$ uniformly at random, and then sum up the corresponding difference queries over $S$ (mapped to an arbitrary ordered list) to recover the estimator $\sum_{T: T \cap S \neq \emptyset} X_T$.
An adversary's best response is then to concentrate all variance on singleton variables $\{X_{\{i\}}\}_{i=1}^n$, since those variables are the most likely not covered by the set of queried experts $S$. 
The approximator's error then corresponds to $\Pr[\ones_{\{i\} \cap S = \emptyset} = 1] = {n-1 \choose d} / {n \choose d} = 1 - \frac{d}{n}$.

\subsection{Bounded agent and order complexity} \label{subsec:agent-order}
Since agent and order complexity put more burden on the experts to calculate complicated query responses, we also aim to understand how limiting these complexity measures to realistic magnitudes affect worst-case error. 
Because query complexity limits both agent and order complexity, an immediate corollary of Theorem~\ref{thm:limited-query-error} is that approximators of limited agent and order complexity inherit $1 - \frac{d}{n}$ as an upper bound on achievable error.
We can then ask whether the principal is able to do any better when they are not constrained by query numbers. 

Throughout the section, we restrict attention to randomized \emph{linear} aggregation rules. 
We believe it is reasonable to study this class under the partial information model, since we previously showed in Proposition~\ref{prop:linear-approximators} that any approximator implemented by linear aggregation rules is a linear function $\sum_T \beta_T X_T$, as is the random variable $Y$ the principal aims to estimate.
We denote $\Lorder$ as the set of approximators implementable using randomized linear aggregation rules restricted to agent and order complexity at most $d$.
We also let $P_d$ be the set of polynomials of degree at most $d$. 

Our main result indeed improves on the error bounds of Theorem~\ref{thm:limited-query-error} in the limited agent and order complexity case.
We characterizes bounds on the best achievable error under approximators in $\Lorder$, where there are three distinct asymptotic regimes.

\begin{theorem} \label{thm:agent-order-error}
    We obtain the following bounds across different regimes of $d$:
    \begin{itemize}
        \item If $d = o(\sqrt{n})$, $\min_{A \in \Lorder} \error(A) = 1 - \Theta\left(\frac{d^2}{n}\right)$.
        \item If $d = \Theta(\sqrt{n})$, $\lim_{n \to \infty} \min_{A \in \Lorder} \error(A) \in (0,1)$.
        \item If $d = \omega(\sqrt{n})$, $\min_{A \in \Lorder} \error(A) = O\left(e^{\frac{-4d}{\sqrt{n}}}\right)$.\footnote{\cite{linial1990approximate} also characterize an improved upper bound of $2^{-\Omega(n)}$ on error when $d = n-1$ by showing that \emph{Krawtchuk polynomials} satisfy equioscillation constraints.}
    \end{itemize}
\end{theorem}

\begin{figure}[t!]
\centering
\begin{subfigure}{0.49\linewidth}
    \centering
    \includegraphics[width=\textwidth]{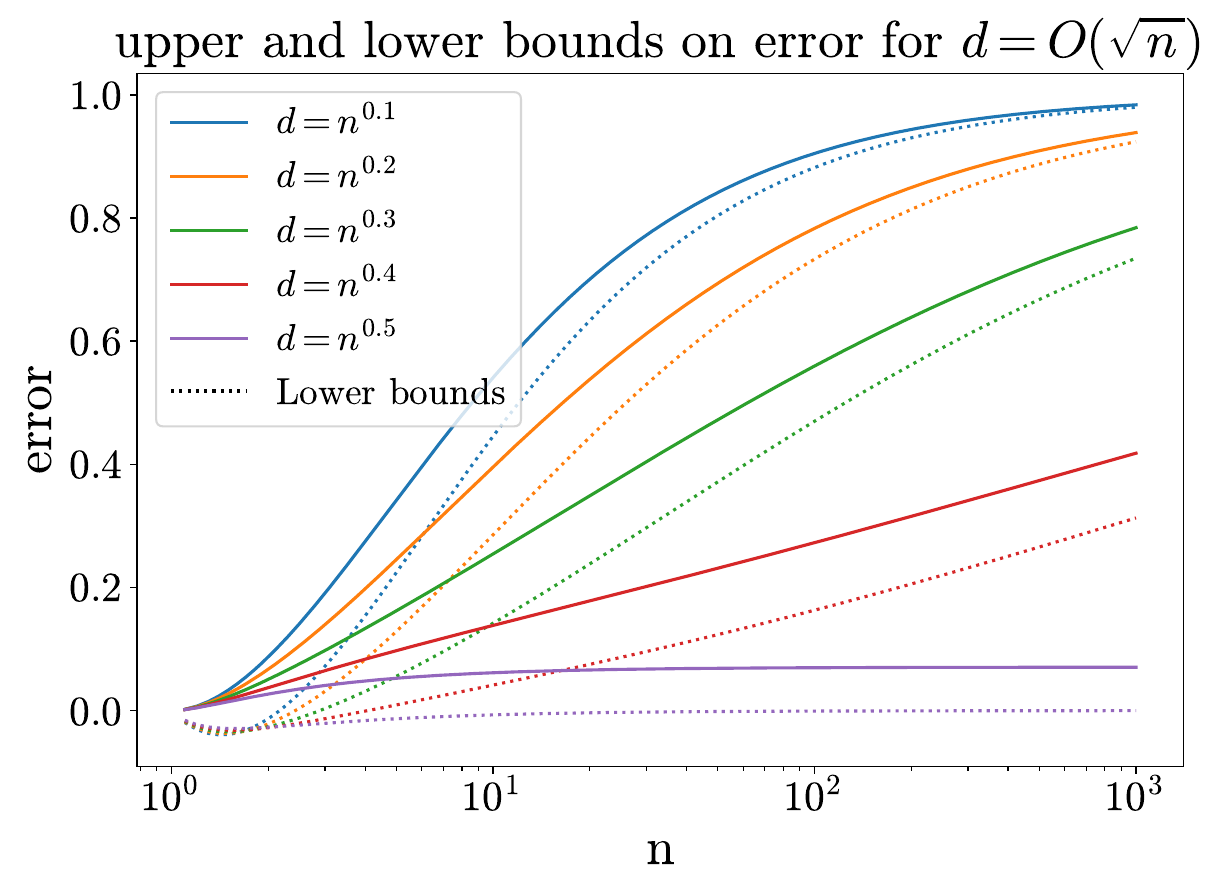}
\end{subfigure}
\begin{subfigure}{0.49\linewidth}
    \centering
    \includegraphics[width=\textwidth]{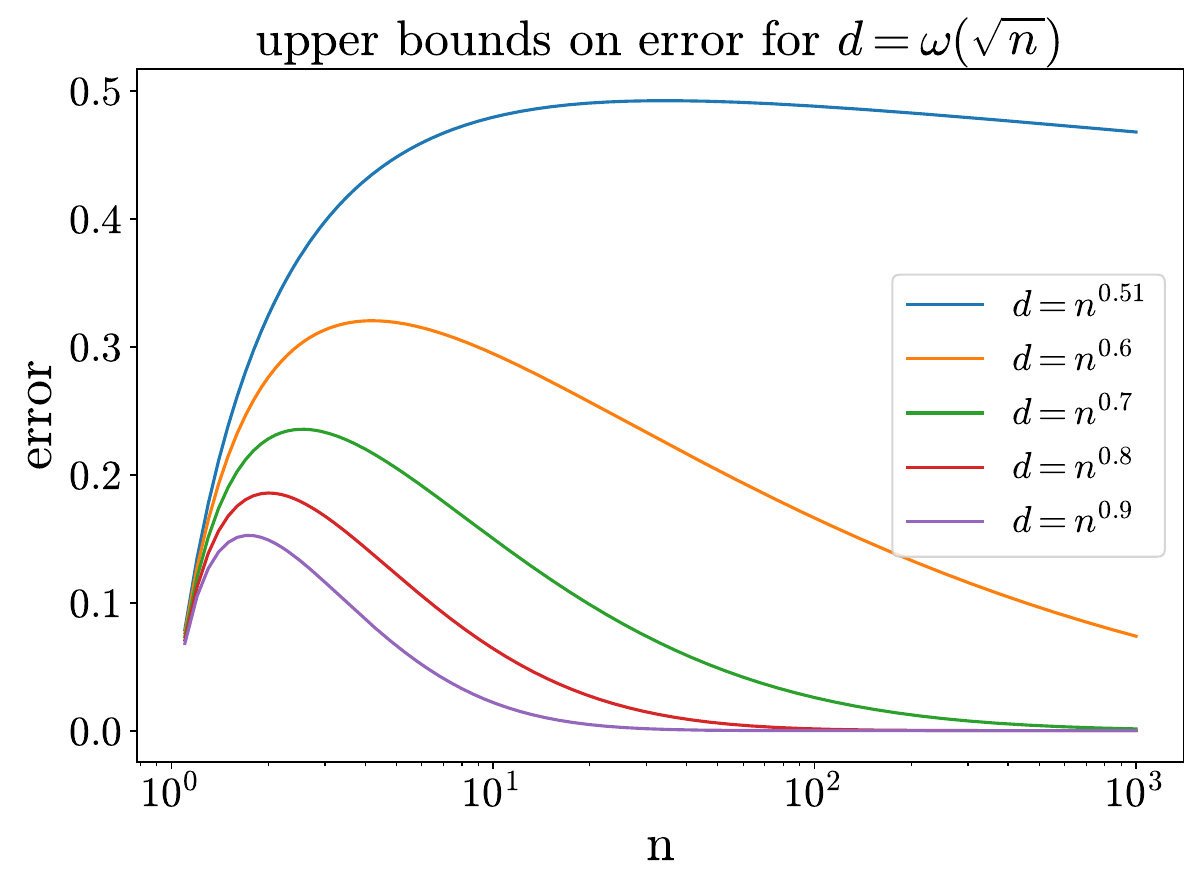}
\end{subfigure}
\caption{
    An illustration of the results stated in Theorem~\ref{thm:agent-order-error}. 
    On the left, we plot upper bounds (filled-in lines) and lower bounds (dotted lines) on error for several settings where $d = O(\sqrt{n})$.
    On the right, we plot upper bounds on error for several settings where $d = \omega(\sqrt{n})$.
    Both graphs are generated using the exact bounds derived in Inequality~\ref{eq:chebyshev-bounds}.
}
\label{fig:error-bounds}
\end{figure}

These results give an almost complete picture of achievable error under limited agent and order complexity.
For $d$ growing smaller than $\sqrt{n}$, error exactly tracks the function $1 - \Theta(\frac{d^2}{n})$, meaning the worst-case error approaches $1$ for a large amount of experts.
Moreover, the quadratic relationship means that there are increasing marginal gains in allowing higher agent and order complexity.
That is, each additional layer of reasoning about other agents or queries increasingly improves error rates (this phenomenon is illustrated in Figure~\ref{fig:error-bounds}).

At $d = \Theta(\sqrt{n})$, there is a phase transition.
In particular, error approaches zero as long as $d$ grows at a rate larger than $\sqrt{n}$; we illustrate this behavior in Figure~\ref{fig:error-bounds}.
We also observe in this regime that there are diminishing marginal gains over $d$ in our upper bounds on error, as the exponential saturates early.
(Note, however, that we are only able to characterize tight upper bounds when $d = \omega(\sqrt{n})$, so that it is feasible better error rates are achievable.)
Thus for reasonable values of $d = \omega(\sqrt{n})$, the tradeoff between error and complexity is quantifiably low, and we can expect high-quality aggregation to be possible \emph{even when restricted by agent and order complexity}.

We note that similar asymptotic trends are observed in the results of~\cite{linial1990approximate}.
Here, the authors study how to approximate the size of the union of $n$ sets, when only intersection sizes for sets of cardinality at most $d$ are available (i.e., the inclusion-exclusion principle cannot be used).
They derive bounds on the worst-case multiplicative error of the estimation, and are able to do so by reducing the error to the same minimax problem that we do for our question. 
The mathematical tools they use to solve the minimax problem are therefore the same as ours.
However, the mapping back to their original problem is quite different, as are the applications of their results. 

While Theorem~\ref{thm:agent-order-error} only characterizes upper bounds on error when $d = \omega(\sqrt{n})$, we are able to prove a weak lower bound on error for all $d$.
Specifically, we can show that it is impossible to achieve optimal error when $d < n$.
\begin{proposition} \label{prop:weak-lower-bound}
    Let $d < n$. 
    Then $\min_{A \in \Lorder} \error(A) > 0$. 
\end{proposition}

\paragraph{Small $d$ regimes.}
We note in particular that we are able to \emph{exactly} characterize the best achievable error when $d=1$, recovering the worst-case results of~\cite{neymanthesis}.
Moreover, we find the exact best achievable error in small $d$ settings as follows.

\begin{proposition} \label{prop:exact-small-d}
    Let $q(n) = \frac{\sqrt{n} + 1}{\sqrt{n} - 1}$. Then
    \[
    \min_{A \in \Lorder} \error(A) = \frac{4}{\left(q(n)^d + q(n)^{-d}\right)^2}
    \]
    if and only if $d = 1$, $d = 2$ and $n$ is odd, or $d = 3$ and $n \equiv 1 \pmod 4$. 
\end{proposition}

\subsection{Proof sketch of Theorem~\ref{thm:agent-order-error}}
We now outline the proof of Theorem~\ref{thm:agent-order-error}.
Tools from approximation theory, specifically \emph{equioscillation theorems} and \emph{Chebyshev polynomials}, play a key role.

\subsubsection{Reducing the class of approximators}
We begin by showing that any approximator $A$ implemented by a randomized linear aggregation rule of limited agent and order complexity $d$ can also be implemented with a randomized \emph{intersection} aggregation rule under the same agent and order complexity bounds.
The proof proceeds by induction over the order of queries, transforming each linear query of order $k$ to an intersection query of order and agent complexity at most $k$.\footnote{We note that Lemma~\ref{lem:linear-intersection}'s proof holds for the general model, not just the partial information model, since one could simply replace the intersection queries with more general iterated expectation queries. }

\mary{Note: I think this also implies that WLOG, $\agentc(R) \leq \orderc(R)$.}
\begin{lemma} \label{lem:linear-intersection}
    Let $A \in \Lorder$.
    Then $A$ is implemented by a randomized intersection aggregation rule $P$ with $\agentc(P) \leq d$ and $\orderc(P) \leq d$.
\end{lemma}

Lemma~\ref{lem:linear-intersection} thus allows us to reduce $\Lorder$ to a simpler class of approximators: namely those implemented by randomized intersection rules $P$ of agent and order complexity at most $d$.
That is, for a fixed $z \sim P$, there exists a vector $\beta \in \reals^{2^d}$ \maryt{technically should be a function of Z...} (WLOG indexed by each subset $S \subseteq [n]$) such that $A = f_z(\Q_{\inter}(d)) = \sum_{S: |S| \leq d} \beta_S \inter(S) = \sum_{S: |S| \leq d} \beta_S \sum_{T \supseteq S} X_T = \sum_T (\sum_{S: |S| \leq d, S \subseteq T} \beta_S) X_T$.
It follows the error of $A$ implemented by $P$ (for a fixed distribution $\D$ and $z \sim P$) can be expressed as
\begin{align}
    \frac{\E(Y - A)^2}{\sum_T \Var_{\D} (X_T)} &= \frac{\E [(\sum_T X_T - \sum_T (\sum_{S: |S| \leq d, S \subseteq T} \beta_S) X_T)^2]}{\sum_T \Var_{\D}(X_T)} \nonumber \\
    &= \frac{\E[(\sum_T g_T(\beta) X_T)^2]}{\sum_T \Var_{\D}(X_T)} \nonumber \\
    &= \frac{\sum_T g_T(\beta)^2 \Var_{\D}(X_T)}{\sum_T \Var_{\D}(X_T)} & \text{(by independence of $X_T$'s),} \label{eq:linear-error}
\end{align}
where $g_T(\beta) \coloneqq 1 - \sum_{S: |S| \leq d, S \subseteq T} \beta_S$.
We can then reduce the class of approximators in $\Lorder$ further when considering bounds on error by using convexity tools.
First, in Lemma~\ref{lem:deterministic-order}, we show that we need only consider \emph{deterministic} intersection aggregation rules that implement approximators in $\Lorder$; and second, in Lemma~\ref{lem:symmetric-functions}, we show we need only consider approximators implemented by \emph{symmetric} intersection aggregation rules, or equivalently coefficients $\beta$ that are symmetric over set sizes. 
We leave proofs to Appendix~\ref{app:agent-order}.

\begin{lemma} \label{lem:deterministic-order}
    For any $A \in \Lorder$, there exists an approximator $A' \in \Lorder$ that is implemented by a deterministic linear aggregation rule such that $\error(A') \leq \error(A)$. 
\end{lemma}

\begin{lemma} \label{lem:symmetric-functions}
    For any approximator $A \in \Lorder$ implemented by an order-$d$ intersection aggregation rule $R$, there is an approximator $A' \in \Lorder$ implemented by an order-$d$ symmetric intersection aggregation rule $\hat R$ such that $\error(A') \leq \error(A)$.
\end{lemma}

\subsubsection{Minimax error and Chebyshev polynomials}
Now that we have pared down the space of aggregation rules that implement approximators in $\Lorder$ significantly, we can directly express the value of the lowest achievable error.
Surprisingly, in Lemma~\ref{lem:intersection-minimax-problem}, we show that the best achievable error corresponds exactly to a simple polynomial approximation problem.
Specifically, the error is the minimum achievable infinity norm on a finite set of integer points, taken over all polynomials of degree at most $d$ that satisfy one linear constraint.

\begin{lemma} \label{lem:intersection-minimax-problem}
    \begin{equation} \label{eq:scaled-discrete-chebyshev}
        \min_{A \in \Lorder} \error(A) = \min_{\substack{p \in P_d \\ p(0) = 1}} \max_{t \in [n]} p(t)^2. 
    \end{equation} 
\end{lemma}

\begin{proof}
By Lemmas~\ref{lem:deterministic-order} and~\ref{lem:symmetric-functions}, we only need to consider \emph{symmetric}, deterministic intersection aggregation rules of order $d$ that implement approximators $A$, or equivalently vectors $\beta \in \reals^d$ with each $\beta_t$ corresponding to subsets of size $t \in [d]$. 
That is, $f_{\beta}(\Q_{\inter}(d)) = \sum_{S: |S| \leq d} \beta_{|S|} \inter(S) = \sum_{S: |S| \leq d} \beta_{|S|} \sum_{T \supseteq S} X_T = \sum_T \left( \sum_{S \subseteq T, |S| \leq d} \beta_{|S|} \right) X_T.$
(We note that adding a nonzero constant to $f_{\beta}(\Q_{\inter}(d))$ only increases the worst-case error.)
Denote $p(T) = \sum_{S \subseteq T, |S| \leq d} \beta_{|S|}$ as the new coefficient for $X_T$.
Then
\[
p(T) =  \sum_{S \subseteq T, |S| = 1} \beta_1 + \sum_{S \subseteq T, |S| = 2} \beta_2 + \ldots + \sum_{S \subseteq T, |S| = d} \beta_d,
\]
or, if we denote $p(|T|)$ instead as a function over the \emph{size} of $T$,
\[
p(|T|) =  \sum_{1 \leq k \leq d} \beta_k {|T| \choose k}.
\]
In other words, $p$ is a polynomial of degree $d$ over the size of $T$ (with no degree 0 term), so that
\[f_{\beta}(\Q_{\inter}(d)) = \sum_T p(|T|) X_T.\]
It follows the error of an approximator $A = f_{\beta}(\Q_{\inter}(d))$ for any symmetric, deterministic intersection aggregation rule $(\Q_{\inter}(d), f_{\beta})$, given a fixed distribution $\D$, is
\begin{align*}
    \frac{\E(Y - A)^2}{\sum_T \Var_{\D}(X_T)} &= \frac{\E (\sum_T X_T - \sum_T p(|T|) X_T)^2}{\sum_T \Var_{\D}(X_T)} \\
    &= \frac{\sum_T \Var_{\D}(X_T)(1 - p(|T|))^2}{\sum_T \Var_{\D}(X_T)}.
\end{align*}

In the worst case over $\D$, an adversary will concentrate all variance on variables $X_T$ such that $t = |T|$ maximizes the value $(1 - f(t))^2$; and WLOG, this variance is the same across each such $X_T$. 
Formally, note for any valid distribution $\D$ that 
    \begin{align*}
    \frac{\sum_T(1 - p(|T|))^2 \Var_{\D}(X_T)}{\sum_T \Var_{\D}(X_T)} &\leq \frac{\max_{t \in [n]} (1 - p(t))^2 \sum_T \Var_{\D}(X_T)}{\sum_T \Var_{\D}(X_T)} = \max_{t \in [n]} (1 - p(t))^2,
    \end{align*}
    with the inequality achieved when $\Var_{\D}(X_T) > 0$ only if $|T| = \arg\max_{t \in [n]}  (1 - p(t))^2$.
    It follows that
\begin{align*}
\min_{A \in \Lorder} \error(A) &= \min_{(\Q_{\inter}(d), f_{\beta})} \error(f_{\beta}(\Q_{\inter(d)})) \\
&= \min_{\substack{p \in P_d \\ p(0) = 0}} \max_{t \in [n]} (1 - p(t))^2 \\
&= \min_{\substack{p \in P_d \\ p(0) = 1}} \max_{t \in [n]} p(t)^2.
\end{align*}
\end{proof}

To achieve tight bounds on $\min_{A \in \Lorder} \error(A)$, then, we just need to characterize a much simpler minimax value over polynomials. 
Luckily, there is a vast literature that studies approximating arbitrary functions with bounded-degree polynomials. 
Specifically, \emph{Chebyshev polynomials} minimize the following problem, which we call the \emph{classic Chebyshev minimax problem}:
\begin{equation} \label{eq:classic-chebyshev}
    \min_{p \in \cM_d} \max_{x \in [-1,1]} |p(x)|,
\end{equation}
where $\cM_d$ is the set of monic polynomials (with a leading coefficient of one) of degree exactly $d$.
The proof makes use of an \emph{equioscillation theorem}: a monic polynomial of degree $d$ achieves the minimax error if and only if the function has $d+1$ points achieving the maximum value $\max_{x \in [-1,1]} |p(x)|$; moreover, these points are alternating in sign.
Intuitively, the polynomial should wiggle above and below zero at the same magnitude as many times as possible under its constraints, to minimize absolute distance to zero.
In fact, the \emph{unique} polynomial that achieves this equioscillation characterization is the $d$-th order Chebyshev polynomial $T_d(x)$, whose extrema occur at points $\cos(\frac{k \pi}{d}), 0 \leq k \leq d$.
% \mary{this is just fun stuff we can move to the appendix :-)} Naively, one may expect the minimizing polynomial to oscillate evenly across the domain; but in fact, Chebyshev polynomials achieve the extreme $p^*$ at points $\cos(\frac{k\pi}{d})$ for $0 \leq k \leq d$, meaning oscillations occur more densely near the endpoints of the interval.
% This phenomenon is related to the \emph{Runge phenomenon}, where polynomials on equispaced interpolation points perform poorly in error.
% This is because polynomials grow slowly at $0$, but experience jumps in magnitude at the endpoints $1$ and $-1$; thus attempting to equally fit points in the center leads to large coefficients and errors at the edges of the interval.

There are two glaring differences between our problem (Equation~\ref{eq:scaled-discrete-chebyshev}) and the classic Chebyshev minimax problem (Equation~\ref{eq:classic-chebyshev}).
First, we are not constrained by monic polynomials of degree exactly $d$, but instead polynomials of degree at most $d$ with a $y$-intercept at zero.
Second, the domain is the finite set $\{1,2,\dots,n\}$ and not the continuous interval $[-1,1]$.
Fortunately, a similar equioscillation theorem (Lemma~\ref{lem:continuous-equioscillate}) holds under our $y$-intercept constraint if we take an affine transformation of the domain from $[-1,1]$ to $[1,n]$, intuitively because the degree of freedom in constructing polynomials is the same. 
The proof structure largely follows previous equioscillation results. 

We can then apply this affine transformation to Chebyshev polynomials, which exactly characterize the minimax value under our constraints over the continuous interval $[1,n]$ (Lemma~\ref{lem:minmax-chebyshev-value}).
These transformed Chebyshev polynomials $\hat T_d$ immediately give us \emph{upper bounds} on our discrete version of the minimax problem.
We are able to derive lower bounds by noting that, as $n \to \infty$, the grid points in $\{1,2,\dots,n\}$ approach the continuous interval $[1,n]$, and thus the discrete minimax value converges to the continuous value achieved in Lemma~\ref{lem:minmax-chebyshev-value}.
Moreover, we can characterize this convergence rate by bounding a degree-$d$ polynomial's worst-case deviation between grid points according to its derivative (Lemma~\ref{lem:chebyshev-lower-bound}).
Specifically, the upper and lower bounds in Lemmas~\ref{lem:minmax-chebyshev-value} and~\ref{lem:chebyshev-lower-bound} are then
\begin{equation} \label{eq:chebyshev-bounds}
    \frac{4(1 - \frac{d^2}{n-1})^2}{(q(n)^d + q(n)^{-d})^2} \leq \min_{A \in \Lorder} \error(A) \leq \frac{4}{(q(n)^d + q(n)^{-d})^2}, 
\end{equation}
where $q(n) = \frac{\sqrt{n} + 1}{\sqrt{n} - 1}$.\footnote{We again note that Lemmas~\ref{lem:minmax-chebyshev-value} and~\ref{lem:chebyshev-lower-bound} are similarly posed and utilized in~\cite{linial1990approximate} to prove bounds for their problem on estimating union sizes for sets. }
The proof of Theorem~\ref{thm:agent-order-error} follows as a straightforward asympototic analsyis of Inequality~\ref{eq:chebyshev-bounds}.
We refer the reader to Appendix~\ref{app:agent-order} for formal statements and proofs of the aforementioned Lemmas~\ref{lem:continuous-equioscillate},~\ref{lem:minmax-chebyshev-value}, and~\ref{lem:chebyshev-lower-bound}, as well as the proof of Theorem~\ref{thm:agent-order-error}.

\subsection{Proof sketch of Propositions~\ref{prop:weak-lower-bound} and~\ref{prop:exact-small-d}}
While we are unable to derive tight lower bounds on error for $d = \omega(\sqrt{n})$, we can show that optimal error is not achievable for $d < n$.
To do so, we first apply an identical equioscillation theorem to our minimax problem over the discrete domain $\{1,2,\dots,n\}$: there exist $d+1$ \emph{integers} in $[1,n]$ that achieve the maximum and oscillate in sign (Lemma~\ref{lem:discrete-equioscillate}).
The proof follows the exact same structure as the continuous version. 
It immediately follows that the error of an approximator in the class $\Lorder$ for $d < n$ must be strictly larger than zero:

\begin{proof}[Proof of Proposition~\ref{prop:weak-lower-bound}]
    Assume not for contradiction, so that by Lemma~\ref{lem:intersection-minimax-problem} there exists a polynomial $p \in P_d$ with $p(0) = 1$ such that $\max_{t \in [n]} |p(t)| = 0$. 
    The polynomial $p$ must therefore also achieve the minimax error over the set $\{p \in P_n: p(0) = 1\}$.
    But Lemma~\ref{lem:discrete-equioscillate} then implies that $p$ must have $n+1$ oscillating maximal points, and thus $n$ roots, which is a contradiction.
\end{proof}

We can also use the discrete equioscillation characterization to show when the transformed Chebyshev polynomials exactly achieve the minimax value for our problem: specifically, when its extrema occur at integer points.
By properties of the function $\cos(x)$, the extrema of our transformed Chebyshev polynomial $\hat T_d$ only align with the integers in the following cases: $d = 1$, $d = 2$ and $n$ is odd, or $d = 3$ and $n \equiv 1 \pmod 4$ (Lemma~\ref{lem:exact-d-bounds}).
Proposition~\ref{prop:exact-small-d}, our exact characterization of worst-case error for $d=1,2,3$, then immediately follows. 

\section{Discussion}

We have initiated the study of worst-case aggregation when the principal may ask queries beyond the typical expectation queries.
We believe this is a rich and fruitful avenue of inquiry, and can see many open directions of interest.

\subsection{Beyond partial information}

Perhaps the most pressing direction is to prove results about the complexity of optimal and approximate aggregation beyond the partial information model studied here.
In particular, our results for the partial information seem to rely heavily on independence of the signals.
What can we do if signals are correlated?
A first attempt could be to allow the signals in the partial information model to be correlated with each other rather than independent.
But in this case, we obtain a fully general model, i.e. the information structure $\D$ could be any joint distribution.
As trivial impossibilities exist for generic settings, we must look for some restriction on the information structures.

One promising approach, following \cite{neyman2022smarter}, is to focus on signals that are \textsl{informational substitutes}.
Intuitively, substitutes (introduced by \cite{borgers2013signals,chen2016informational}) are cases where the marginal value of a signal $X_i$ for predicting $Y$ is decreased when other signals $X_j$ are already available.
\cite{neyman2022smarter} showed that under ``weak substitutes'', there is an aggregation rule on $(Y_1,\dots,Y_n)$ achieving error $1-\tfrac{1}{n}$, and this is tight; in other words, one might as well pick a random expert's opinion.
Under a stronger notion of ``projective substitutes'', \cite{neyman2022smarter} showed that better error is possible, but still in the $1-O(\tfrac{1}{n})$ regime.

\subsection{Knowledge of the prior}

Throughout we have assumed that the prior joint distribution $\D$ is unknown to the aggregator.
It is natural to ask if the guarantees we have established can be improved if $\D$ is known.
It turns out the answer is no, at least in the specific questions we study here, as in each case the worst-case instance $\D$ can be determined and known ahead of time.
In other words, evidently minimax duality holds in this setting: $\inf_A \sup_\D \error(A_\D) = \sup_\D \inf_A \error(A_\D)$.

This duality can break beyond the partial information model, however.
For example, if we extend the partial information model to allow non-zero means, i.e.\ $\mu_T := \E[X_T] \neq 0$, then the means become nontrivial to contend with.
\cite[Theorem 3.3]{neyman2022smarter} give an example with $X_{\{i\}} \sim N(\mu,1)$ for all $i$, and $X_T=0$ otherwise, where for agent/order complexity $1$ the optimal aggregation error is $1 - \tfrac 2 n + \tfrac 1 {n^2}$; with knowledge of the prior, one can achieve optimal aggregation (error 0).
Interestingly, all of our constructions for optimal aggregation appear to go through even in this extension, as both the inclusion-exclusion and sum of difference queries automatically cancel out the $\mu_T$ terms.

\subsection{Adapting to easy instances}

\raf{I'm not sure how much detail we want to go into here; probably should skim down considerably and just focus on the example}

We have studied \emph{worst-case} performance guarantees of the form $\error(A) \leq c$, which hold for all instances $\D$ within a model class, like the partial information model.
If this bound is tight, then $\error(A) = c$ for some set of ``hard'' instances.
In some sense, therefore, the worst-case optimal aggregation rules above are designed in response to these hardest instances, and these alone.
With a nod to the highly successful body of work on online learning algorithms for easy data~\citep{vanErven2021metagrad,koolen2015squint,vanErven2016combining}, we would like aggregation rules that retain the same worst-case error, but can improve on easy instances.
% In online learning, such guarantees can take the form of $O(\sqrt{L^*})$ regret, where $L^*$ is the loss of the best expert; notably $L^*$ is at most the number of rounds $T$, but could be much lower.
% % are akin to the long line of literature on online learning for easy data, where algorithms seek lower regret on easy sequences (like all 1's for binary prediction), even while retaining the same $\sqrt{T}$ regret bound over all possible sequences.
% In our setting, guarantees might take the form $\error(\hat A) \leq c - f(\D)$, where $f(\D) \geq 0$ measures how easy the instance is.
% As $f \geq 0$, this bound would still imply the worst-case bound.

To illustrate, consider a generalization of Example~\ref{ex:common} where the private and common signals $X_1,\ldots,X_{n+1}$ are still independent and mean-zero but could be any distribution otherwise.
For the optimal aggregation rule (over posterior expectation queries) proposed in~\cite{neymanthesis}, $A = f(Y_1,\dots,Y_n) = \frac{2}{n+1}\sum_{i=1}^n Y_i$, its worst-case error $\error(A) = 1 - \tfrac{4n}{(n+1)^2}$ occurs for a distribution with $\Var(X_1) = \cdots = \Var(X_{n+1})$.
On this instance, therefore, it would be pointless to ask experts for other statistics like $Q_{i,2} = \Var[Y\mid \cS_i]$.

But such variance information would help significantly \emph{outside} the worst-case.
Specifically, as $Q_{i,2} = \sum_{j\neq i} \Var(X_i)$, we can solve for $\Var(X_1),\ldots,\Var(X_n)$ and thus derive better weights for the $Y_i$.
In particular, the optimal weight $\alpha_i$ for $Y_i$ will now be
\begin{align*}
  \alpha_i =  1 - n \left(\frac{n-1}{n+1}\right) \frac{\tau_i}{\tau}~,
\end{align*}
where $\tau_i = 1/\Var(X_i)$ is the ``precision'' of $X_i$, and $\tau = \sum_{i=1}^n \tau_i$.
These weights are still
% average to $\frac{2}{n+1}$, and indeed will be
$\frac{2}{n+1}$ when the variances are all equal, giving the worst-case performance bound.
Outside of that worst case, however, performance strictly improves.
Specifically, fixing the total variance, the performance is now strictly increasing in the total precision $\tau$.

\bibliographystyle{alpha}
\bibliography{citations,citations2,waggoner,raf-papers,agg}

\break

\appendix

\section{Omitted proofs}

\subsection{Proofs from \S~\ref{sec:optimal}} \label{app:optimal}

We elaborate on the computation of the approximator implemented by a specific linear aggregation rule here:
\begin{examplecon}{\ref{ex:linear}}
    Note that by definition,
    \begin{align*}
        Q_3 &= \E[ - \E[ Y \mid \cS_1] + 2 \E[ Y \mid \cS_2] \mid \cS_3] \\
        &= \E[ - \sum_{T \ni 1} X_T + 2\sum_{T \ni 2} X_T \mid \cS_3] \\
        &= - \sum_{T \ni 1,3} X_T + 2 \sum_{T \ni 2,3} X_T.
    \end{align*}
    It follows that
    \begin{align*}
        f(\Q) &= Q_1 + Q_2 + Q_3 \\
        &= - \sum_{T \ni 1,3} X_T + 2 \sum_{T \ni 2,3} X_T + \sum_{T \ni 1} X_T + \sum_{T \ni 2} X_T \\
        &= \sum_{T \subseteq [n]} X_T,
    \end{align*}
    with values of $\beta_T$ given in the main body for Example~\ref{ex:linear}.
\end{examplecon}

\begin{proof}[Proof of Theorem~\ref{thm:opt-intersection}]
    We define a deterministic intersection aggregation rule ($\Q_{\inter}(n), f)$ that implements an approximator $A$ with $\error(A) = 0$.
    Then by Proposition~\ref{prop:inter-query-order}, $\Q_{\inter}(n)$ is DAG-elicitable with agent and order complexity at most $n$, and query complexity $2^n$.
    We define $f$ as
    \begin{align*}
        f(\Q_{\inter}(n)) &= \sum_{\substack{S \subseteq [n] \\ S \neq \emptyset}} \left( \sum_{\substack{T \subseteq S \\ T \neq \emptyset}} (-1)^{|S \backslash T|} \right) \inter(S) \\
        &= \sum_{\substack{T \subseteq [n] \\ T \neq \emptyset}} \sum_{S \supseteq T} (-1)^{|S \backslash T|} \inter(S) .
    \end{align*}
    Now, note that for each nonempty $T \subseteq [n]$,
    \begin{align*}
    \sum_{S \supseteq T} (-1)^{|S \backslash T|} \inter(S) = X_T;
    \end{align*}
this identity holds by the Möbius inversion formula over the Boolean lattice, or equivalently, a generalization of the inclusion–exclusion principle applied to the function $g(T) \coloneqq X_T$.

Thus we have $f(\Q_{\inter}(n)) = \sum_{T \subseteq [n]} X_T = Y$.
Given the approximator $A$ that $(f,\Q)$ implements, it follows that $\error(A) = 0$ with probability one. 
\end{proof}

\begin{proof}[Proof of Proposition~\ref{prop:diff-query-order}]
    Let $d = |L|$.
    Denote $Q_{i_k} \coloneqq \diff(i_1,i_2,\dots,i_k)$, so that we have $\Q_{\diff}(L) = (Q_{i_1}, Q_{i_2}, \dots, Q_{i_d})$.
    Then by definition, since $Q_{i_k} = \E[Y - \sum_{\ell < k} Q_{i_{\ell}} \mid \cS_k]$, $Q_{i_k}$ is elicited by the payment $\pi(Q, \{Q_{i_{\ell}}\}_{\ell < k} \cup \{ Y \}) = 1 - (Q - (Y - \sum_{\ell < k} Q_{i_{\ell}}))^2$ for $i > 1$; $Q_{i_1} = Y_{i_1}$ is elicited by $\pi(Q, Y) = 1 - (Q - Y)^2$.
    Thus we can define a graph $G = (V,E)$ that DAG-implements $\Q_{\diff}(L)$ with $V \coloneqq \Q_{\diff}(L) \cup \{ Y \}$ as follows: $(Q_{i_1}, Y) \in E$ and $(Q_{i_k}, Q_{i_{\ell}}) \in E$ for all $1 < k \leq n$ and $\ell < k$. 
    Since all other queries are reachable from $Q_{i_d}$ in $G$, and the longest path in $G$ from $Q_{i_d}$ to $Y$ is of length $n$, $\agentc(\Q_{\diff}(L)) \leq d$ and $\orderc(\Q_{\diff}(L)) \leq d$.
    Meanwhile, $\queryc(\Q_{\diff}(L)) = |\Q_{\diff}(L)| = d$.
\end{proof}

\subsection{Proofs from \S~\ref{subsec:query}} \label{app:query}

Here, we prove both the upper and lower bounds which immediately imply the statement of Theorem~\ref{thm:limited-query-error}.

\begin{lemma} \label{lem:limited-query-lower-bound}
    Let $d < n$. Then
    $\min_{A \in \Aquery} \error(A) \geq 1 - \frac{d}{n}$.
\end{lemma}

\begin{proof}
    By the minimax inequality,
    \begin{align*}
        \min_{A \in \Aquery} \error(A) &= \min_{A \in \Aquery} \max_{\D} \frac{L(A)}{\sum_T \Var_{\D}(X_T)} \\
        &\geq \max_{\D} \min_{A \in \Aquery} \frac{L(A)}{\sum_T \Var_{\D}(X_T)} \\
        &= \max_{\D} \min_{\substack{(f, \Q) \in \DAGset \\ \queryc(\Q) \leq d}} \frac{L(f(\Q))}{\sum_T \Var_{\D}(X_T)}.
    \end{align*}
    The last equality, where we need only consider deterministic aggregation rules that implement $A$, holds since for any function $\ell: \reals \to \reals$, $\E_{Z \sim P} \ell(f_Z(\Q_Z)) \geq \min_{z \in \supp(P)} \ell(f_z(\Q_z))$ (and the inequality is realized by the point mass on the argmin pair $(f_Z, \Q_Z)$).

    Since $\queryc(\Q) \leq d$, by definition $\agentc(\Q) \leq d$ and $\orderc(\Q) \leq d$. 
    (The maximum path length from any query to the sink in any graph $G$ that DAG-implements $\Q$ is upper bounded by $|\Q|$, as is the number of agents reachable from any query.)
    It immediately follows that in the best case, with a deterministic set of queries, the principal can obtain knowledge of $\left(\bigcup_{i \in S} \cS_i\right) \cup \D $ for some subset $S$ of agents with $|S| \leq d$.
    
    Consider the adversarial distribution $\D$ with all variance concentrated on the singleton sets. 
    That is, $\Var_{\D}(X_{\{i\}}) = 1$ for all $i \in [n]$, and $\Var_{\D}(X_T) = 0$ for all $T \subseteq [n]$, $|T| > 1$. 
    It follows that $Y = \sum_{i=1}^n X_{\{i\}}$.
    By properties of squared loss, $\E_{\D} [Y \mid \bigcup_{i \in S} \cS_i ] = \sum_{i \in S} X_{\{i\}}$ minimizes $\E_{\D} [ (Y - A)^2 \mid \bigcup_{i \in S} \cS_i]$.
    Thus
    \begin{align*}
         \min_{\substack{(f, \Q) \in \DAGset \\ \queryc(\Q) \leq d}} \frac{\E(Y - f(\Q))^2}{\sum_T \Var_{\D}(X_T)} &\geq \frac{\E [(\sum_{i=1}^n X_{\{i\}} - \sum_{i \in S} X_{\{i\}})^2]}{n} \\
         &= \frac{\sum_{i \in [n] \backslash S} \Var_{\D} [X_{\{i\}}]}{n} \\
         &= 1 - \frac{d}{n}.
    \end{align*}
\end{proof}

\begin{lemma} \label{lem:limited-query-upper-bound}
    Let $d < n$. Then there exists an approximator $A \in \Aquery$ such that $\error(A) = 1 - \frac{d}{n}$.
\end{lemma}

\begin{proof}
    Let $\T_d \coloneqq \{S \subseteq [n]: 1 \leq |S| \leq d\}$ be the collection of expert subsets of size at most $d$. 
    For a subset $S \subseteq [n]$, let $L(S)$ map $S$ to an arbitrarily ordered list $L(S)$.

    We define the randomized aggregation rule $P \coloneqq \text{unif}(\T_d)$.
    In other words, $P$ uniformly randomly samples a subset $S$ of $d$ agents.
    Then we let $(\Q_S, f_S) \coloneqq (\Q_{\diff}(L(S)), \sum_{Q \in \Q_S} Q)$, i.e. once a random subset of experts $S$ is chosen and mapped to an ordered list $L(S)$, the randomized aggregation rule outputs the linear combination of the corresponding difference queries.
    Let $A_P$ be the approximator that $P$ implements. 
    By Proposition~\ref{prop:diff-query-order}, for any list $L$ of length $d$, $\Q_{\diff}(L)$ has agent and order complexity $d$; moreover, the query complexity is $d$. 
    It follows that $\queryc(P) = d$, $\orderc(P) \leq d$, and $\agentc(P) \leq d$, so that $A_P \in \Aquery$.
    
    For a fixed subset $S$, by Proposition~\ref{prop:diff-queries}, 
    \[
    f_S(\Q_{L(S)}) = \sum_{k=1}^d \sum_{\substack{T: i_k \in T, \\ i_{\ell} \notin T \; \forall \ell < k}} X_T = \sum_{T: T \cap S \neq \emptyset} X_T.
    \]

    Let $I_T \coloneqq \ones_{T \cap S = \emptyset}$. 
    Then the error of $A_P$ under a fixed distribution $\D$ is

    \begin{align*}
        \frac{\E_{S, \D} (\sum_T X_T - f_S(\Q_{L(S)}))^2}{\sum_T \Var (X_T)} &= \frac{\E_{S,\D} [(\sum_T X_T - \sum_{T: T \cap S \neq \emptyset} X_T)^2]}{\sum_T \Var (X_T)} \\
        &= \frac{\E_{S \sim \text{unif}(\T_d)} \E_{\{X_T\} \sim \D}[(\sum_{T: T \cap S = \emptyset} X_T)^2]}{\sum_T \Var(X_T)} \\
        &= \frac{\sum_T \sum_U \E[X_T X_U] \Pr[I_T = 1, I_U = 1]}{\sum_T \Var(X_T)} \\
        &= \frac{\sum_T \beta_T \Var(X_T)}{\sum_T \Var(X_T)} & \text{(by independence of $X_T$'s),} 
    \end{align*}
    where $\beta_T \coloneqq \Pr[I_T = 1] = \frac{{n - |T| \choose d}}{{n \choose d}}$.
    In particular, if $|T| = 1$, $\beta_T = \frac{{n - 1 \choose d}}{{n \choose d}} = 1 - \frac{d}{n}$.
    We note that $\beta_T$ is the largest for $T$ of size $1$, so that in the worst case over $\D$ an adversary will put all variance on the singleton sets.
    Formally, note for any valid distribution $\D$ that 
    \begin{align*}
    \frac{\sum_T \beta_T \Var_{\D}(X_T)}{\sum_T \Var_{\D}(X_T)} &\leq \frac{\beta_1 \sum_T \Var_{\D}(X_T)}{\sum_T \Var_{\D}(X_T)} = \beta_1,
    \end{align*}
    with the inequality achieved when $\Var_{\D}(X_T) > 0$ only if $|T| = 1$.
    Thus,
    \begin{align*}
        \error(A_P) = \max_{\D} \frac{\sum_T \beta_T \Var_{\D}(X_T)}{\sum_T \Var_{\D}(X_T)} \leq \beta_1 &= 1 - \frac{d}{n}.
    \end{align*}
\end{proof}

\subsection{Lemmas and proofs from \S~\ref{subsec:agent-order}} \label{app:agent-order}

\begin{proof}[Proof of Lemma~\ref{lem:linear-intersection}]
    Let $P$ be the randomized linear aggregation rule that implements $A$, with $\agentc(P) \leq d$ and $\orderc(P) \leq d$.
    For any $Z \sim P$, let $R_Z = (\Q_Z, f_Z)$ be a deterministic linear aggregation rule, where $f_Z(\Q_Z) = \sum_{Q \in \Q_Z} \beta_Q Q$, $\agentc(R) \leq d$, and $\orderc(R) \leq d$.
    Note first that the order complexity of each $Q \in \Q_Z$ is at most $d$ by definition under some graph $G$ that DAG-implements $R_Z$.
    
    Our proof uses induction to iteratively transform $\Q_Z$ into a new set of queries $\Q_{Z_k}$, where we let $Z_k \coloneqq g_k(Z)$ be a function of $Z$.
    At step $k$, we will show by induction that for any linear query $Q \in \Q_Z$ of order complexity at most $k$ in $G$, we can re-write the query as \maryt{differentiate between query output and function here...}$Q = \sum_{S : 1 \leq |S| \leq k} c_{S,Q} \inter(S)$ for some set of coefficients $\{ c_{S,Q} \}$.
    We can then replace the query set $\Q_{Z_{k-1}}$ with $\Q_{Z_k} \coloneqq (\Q_{Z_{k-1}} \backslash \{ Q \}_{\orderc(Q) \leq k}) \cap \Q_{\inter}(k)$.
    It follows that $\Q_{Z_d} = \Q_{\inter}(d)$.
    We can thus define an intersection aggregation rule $(\Q_{Z_d}, f_{Z_d})$, where $f_{Z_d}$ is a linear function over $\Q_{Z_d}$ with
    \begin{align*}
        f_{Z_d}(\Q_{Z_d}) &= \sum_{S: |S| \leq d} \left(\sum_{Q \in \Q_Z} \beta_Q c_{S,Q}\right) \inter(S) \\
        &= \sum_{Q \in \Q_Z} \beta_Q \sum_{S : |S| \leq d} c_{S,Q} \inter(S) \\
        &= \sum_{Q \in \Q_Z} \beta_Q Q \\ 
        &= f_Z(\Q_Z) & \text{a.s.}
    \end{align*}

    Let $Z_d \sim P'$, where $P'$ pushes forward $P_Z$ under $g_d(Z)$.
    Since the approach holds for any fixed $Z$, it follows that with probability one, $f_{Z_d}(\Q_{Z_d}) = A$; and since $\Q_{Z_d} \coloneqq \Q_{\inter}(d)$, by Proposition~\ref{prop:inter-query-order} $\agentc(P') \leq d$ and $\orderc(P') \leq d$.
    The statement follows. 

    We proceed by proving the inductive statement: each query $Q \in \Q_Z$ of order at most $k$ can be written as $Q = \sum_{S : 1 \leq |S| \leq k} c_{S,Q} \inter(S)$ for some set of coefficients $\{ c_{S,Q} \}$.
    For the base case, take $k = 1$.
    Then a linear query of order $1$ can only be scored against the sink, $Y$. 
    That is, if $\agent(Q) = i$, $Q = \E\left[ \alpha_{Q, Y} Y \mid \cS_i\right] = \sum_{T \ni i} \alpha_{Q,Y} X_T = c_{\{i\},Q} \inter(\{i\})$ for $ c_{\{i\},Q} \coloneqq \alpha_{Q,Y}$.

    Now, assume each $Q' \in \Q_Z$ of order at most $k$ can be written as $Q' = \sum_{S : 1 \leq |S| \leq \ell} c_{S,Q} \inter(S)$ for some set of coefficients $\{ c_{S,Q} \}$.
    Take a linear query $Q = \sum_{Q' \in N_G(Q)} \alpha_{Q,Q'} Q'$ of order complexity at most $k+1$, being asked of agent $i$.
    Note by construction that each $Q' \in N_G(Q)$ must be of order at most $k$.
    Thus, by the inductive hypothesis, each $Q' = \sum_{S: 1 \leq |S| \leq k} c_{S,Q'} \inter(S)$.

    Since agent $i$ is answering $Q$ and each $X_T$ is independent with mean zero, only $X_T$ with $i \in T$ survive in the final expression. 
    Thus 
    \[
    Q = \sum_{Q'}  \alpha_{Q,Q'} \sum_{S: |S| \leq k} c_{S,Q'} \inter(S \cup \{i\}).
    \]
    Now for $S' = S \cup \{i\}$, $|S'| = |S| + 1 \leq k + 1$.
    Define 
    $c_{S',Q} = \sum_{Q' \in N_G(Q)} \sum_{S: |S| \leq k, S \cup \{i\} = S'} \alpha_{Q'} c_{S,Q'}$.
    (If the sum is over an empty set, then $c_{S,Q'} = 0$.)
    Then we have 
    \[ 
    Q = \sum_{S': |S'| \leq k + 1} c_{S',Q} \inter(S').
    \]

\end{proof}

\begin{proof}[Proof of Lemma~\ref{lem:deterministic-order}]
    Let $A \in \Lorder$, so there exists a randomized linear aggregation rule $P$ with $\orderc(P) \leq d$, $\agentc(P) \leq d$ that implements $A$. 
    By Equation~\ref{eq:linear-error}, a randomized linear aggregation rule $P$ is equivalent to a distribution $\hat P \in \Delta(\reals^{2^d})$, with
    \[ \error(\hat P) = \E_{\beta \sim \hat P} \max_{\substack{v \in \reals^{\geq 0} \\ v \neq 0}} \frac{G(\beta, v)}{\sum_T v_T},\]
    where $G(\beta, v) = \sum_T (1 - \sum_{S: |S| \leq d, S \subseteq T} \beta_S)^2 v_T$. 
    Note that $G$ is convex in the vector $\beta$, and the maximum over convex functions is also convex. 
    Formally, for a subset $T$ let vector $u_T \in \reals^{2^d}$ satisfy $(u_T)_S = 1$ if $S \subseteq T$, and 0 otherwise. 
    Then $\sum_{S: |S| \leq d, S \subseteq T} \beta_S = u_T^{\top} \beta$, and $G(\beta, v) = \sum_T v_T (1 - u_T^{\top} \beta)^2$.
    It follows by Jensen's inequality that
    \[ \error(\hat P) = \E_{\beta \sim \hat P} \max_{\substack{v \in \reals^{\geq 0} \\ v \neq 0}} \frac{G(\beta, v)}{\sum_T v_T} \geq \max_{\substack{v \in \reals^{\geq 0} \\ v \neq 0}} \frac{f(\E_{\hat P} \beta, v)}{\sum_T v_T} = \error(R), \]
    where \maryt{todo: identify $A$ in intersection query roles according to $\beta$}
    $R = (\Q_{\inter}(d), f)$ with $f(\Q_{\inter}(d)) = \sum_{S: |S| \leq d} \beta_S \inter(S)$ and each $\beta(S) \coloneqq \E_{\hat P} \beta$.
    The statement follows by considering the approximator that $R$ implements, since $R$ is an intersection aggregation rule of order $d$.
\end{proof}

\begin{proof}[Proof of Lemma~\ref{lem:symmetric-functions}]
    Let $A$ be implemented by an order-$d$, deterministic intersection aggregation rule $R = (\Q_{\inter}(d), f)$. 
    Take some set of coefficients $\beta \in \reals^{2^d}$, identifying the aggregation function $f$.
    As shown in Lemma~\ref{lem:deterministic-order}, \[\error(A) = \max_{\substack{v \in \reals^{\geq 0} \\ v \neq 0}} \frac{f(\beta, v)}{\sum_T v_T} \eqqcolon F(\beta)\] is a convex function over $\beta$.
    Note also that the function $F(\beta)$ is invariant to permutations over the sets.
    In particular, take some permutation $\sigma$ over $[n]$, which operates on a set $S$ as $\sigma(S) = \{\sigma(i): i \in T\}$ (i.e. it maps the set to a different set of agents with the same cardinality).
    Let $\beta_T^{\sigma} = \beta_{\sigma(T)}$. 
    \maryt{TODO: elaborate here}
    Then note that $F(\beta) = F(\beta^{\sigma})$.

    Let $\text{Sym}(n)$ be the set of permutations on $[n]$.
    Now, if we just take the \emph{average} of the coefficients over all permutations, i.e. $\beta' = \frac{1}{n!} \sum_{\sigma \in \text{Sym}(n)} \beta^{\sigma}$ (with element-wise addition and division), by Jensen's inequality 
    \[ 
    F(\beta') \leq \frac{1}{n!} \sum_{\sigma \in \text{Sym}(n)} F(\beta^{\sigma}) = F(\beta).
    \]

    In other words, for each subset with the same length, we re-assign it the \emph{average} of all coefficients for that length.
    It follows for any non-symmetric assignment $\beta$, there exists a symmetric assignment $\beta'$ of coefficients that does not increase the error. 
    Moreover, $F(\beta') = \error(A')$ for $A'$ the approximator that $R' = (\Q_{\inter}(d), \hat f)$ implements, with $\hat f$ the linear function identified by $\beta'$.
    The statement follows. 
\end{proof}

\begin{lemma} \label{lem:continuous-equioscillate}
    \[p^* \in \arg\min_{\substack{p \in P_d \\ p(0) = 1}} \max_{t \in [1,n]} |p(t)|\] if and only if $p^*$ equioscillates at $d+1$ points in $[1,n]$ (i.e. there exist $d + 1$ values $x_i \in \reals$ such that $1 \leq x_0 \leq \ldots \leq x_d \leq n$, where $|p^*(x_i)| = \sigma(-1)^i \max_{x \in [1,n]} p^*(x)$ with $\sigma$ either $-1$ or $+1$). 
    Moreover, the polynomial $p^*$ is unique.
\end{lemma}

\begin{proof}
    Our proof tracks the same reasoning as the original equioscillation theorem; however, since our polynomial constraints and domain are different, we provide a proof with adjustments here. 
    We refer the reader to Theorem 8.4 in~\cite{suli2003introduction}, whose structure we largely follow.
    \paragraph{Sufficiency.}
    Assume if a function $p^*$ satisfies the equioscillation property, it does not achieve the minimax value, so that there is some other $p \in P_d$ with $p(0) = 1$ such that $\max_{x \in [1,n]} |p(x)| <  \max_{x \in [1,n]} |p^*(x)|$.
    Consider the polynomial $h = p^* - p$, which is also of degree at most $d$. 
    First note that $p^*(0) = p(0) = 1$, so that $h(0) = 0$.
    Moreover, for every $x_i$ of the $d+1$ points $x_0, \ldots, x_d$, we know by assumption that $|p(x_i)| < |p^*(x_i)|$, meaning $h(x_i) = p^*(x_i) - p(x_i)$ must have the same sign as $p^*$.
    Thus $h$ will alternate in sign on these $d+1$ points, meaning $h$ has $d+2$ roots.
    This is a contradiction, since $h$ is a polynomial of degree at most $d$. 
    \paragraph{Necessity.}
    Assume $p^* \in P_d$ achieves the minimax objective, and let $\X$ be the set of points in $[1,n]$ such that $|p^*(x)| = \max_{x \in [1,n]} |p^*(x)|$.
    We first show that, for all $q \in P_d$ with $q(0) = 0$, $\max_{x \in \X} p^*(x) q(x) \leq 0$. 
    Suppose not, i.e. there exists some $\tilde q \in P_d$, $q(0) = 0$ such that $\max_{x \in \X} p^*(x) q(x) > 2\epsilon$. 
    By continuity of the function $p^*(x) q(x)$, we can find a $\delta > 0$ such that 
   \begin{equation} \label{eq:neighborhood-eps}
   \max_{x \in \tilde \X} p^*(x) q(x) > \epsilon,
   \end{equation}
    where $\tilde X \coloneqq \{ t \in [a, b]: \min_{x \in \X} |t - x| < \delta\}$ (i.e. $\tilde X$ is the set of points in a small neighborhood around those in $\X$).
    Let $E \coloneqq \max_{x \in [1,n]} |p^*(x)|$, and let $\tilde E \coloneqq  \max_{x \in [1,n]} |\tilde q(x)|$.
    We now design a function $\tilde p$ that achieves a better minimax value, with $\tilde p(x) = p^*(x) - \lambda \tilde q(x)$ and $\lambda \in (0, \min\{2\epsilon/M^2, \eta/\tilde E\})$.
    This leads to a contradiction. 
    We note that for all $x \in \tilde X$,
    \begin{align*}
        \tilde p(x)^2 &= p^*(x)^2 - 2 \lambda \tilde q(x) p^*(x) + \lambda^2 \tilde q(x)^2 \\
        &= E^2 - 2 \lambda \tilde q(x) p^*(x) + \lambda^2 \tilde q(x)^2 \\
        &< E^2 + 2\lambda \epsilon + \lambda^2 M^2 & \text{(by Equation~\ref{eq:neighborhood-eps})} \\
        &\leq E^2 - \frac{4\epsilon^2}{M^2} + \frac{4\epsilon^2}{M^2} \\
        &= E^2.
    \end{align*}
    Meanwhile, consider the set $\hat \X \coloneqq [1,n] \backslash \tilde \X$.
    By continuity, there is a value $\eta > 0$ such that $|p^*(x)| \leq E - \eta$ for all $x \in \hat \X$.
    Since $\lambda < \eta / \tilde E$, we have for all $x \in \hat X$ that
    \begin{align*}
        |\tilde p(x)| &= |p^*(x) - \lambda \tilde q(x)| \\
        &\leq |p^*(x)| + \lambda | \tilde q(x)| \\
        &< (E - \eta) +  \eta\\
        &\leq E.
    \end{align*}

    Now we need to show that the set $\X$ contains $d+1$ points, and the error oscillates in sign at these points.
    Suppose this is not the case, so we can only identify $m$ points with oscillating sign in $\X$ with $1 \leq m < d + 1$, labeled $x_0 < \dots < x_{m-1}$.
    We will then construct a function $q \in \P_d$, $q(0) = 0$ that violates the previous statement., i.e. $\max_{x \in \X} p^*(x) q(x) > 0$. 
    Note first that $E \coloneqq \max_{x \in [1,n]} |p^*(x)| > 0$, since $p^*$ is a polynomial of degree at least $1$.
    If $m = 1$, then $p^*$ does not change sign on $\X$. 
    If we then take $q(x) = x$, either $\max_{x \in \X} q(x) p^*(x) > 0$ or $\max_{x \in \X} -q(x) p^*(x) > 0$.
    If $m > 1$, there exist $m-1$ zeroes of $p^*$, $\tilde x_0 < \dots < \tilde x_{m-2}$, with each $\tilde x_i \in (x_{i-1}, x_i)$ (where we define $x_{-1} \coloneqq 1$ and $x_m \coloneqq n$).
    Now we construct the function $q(x) = x \prod_{j=0}^{m-2} (x - \tilde x_j)$.
    First note that $q(0) = 0$, and the degree of $q$ is $m \leq d$.
    Since the sign of $q$ does not change between its roots, it does not change in the intervals $(1, \tilde x_0), (\tilde x_0, \tilde x_1), \dots, (\tilde x_{m-1}, n)$; and it flips between these intervals.
    It follows that the sign of $q(x) p^*(x)$ is the same for all $x \in \X$, and either $\max_{x \in \X} q(x) p^*(x) > 0$ or $\max_{x \in \X} -q(x) p^*(x) > 0$.

    \paragraph{Uniqueness.}
    Suppose by contradiction there are two polynomials $p_1, p_2 \in \P_d$ that achieve the minimax value with $p_1 \neq p_2$, $p_1(0) = 1$, and $p_2(0) = 1$.
    Then take $q(x) = \frac{p_1(x) + p_2(x)}{2}$, and note that $q \in \P_d$ and $q(0) = 1$.
    Moreover, applying the triangle inequality, one can see that $q(x)$ also achieves the minimax value on $[1,n]$.
    By the necessary condition proven above, $q$ has $d+1$ points $(x_0, \dots, x_{d+1})$ in $[1,n]$ that oscillate in sign and achieve the extreme value. 
    We can then show for each $x_i$ that $p_1(x_i) = p_2(x_i)$.
    In particular, note $2 |q(x_i)| = |p_1(x_i) + p_2(x_i)| = 2E$, $|q(x_i)| \leq E$, and $|p_2(x_i)| \leq E$.
    We therefore must have $(p_1(x_i) - p_2(x_i))^2 = 2 p_1(x_i)^2 + 2 p_2(x_i)^2 - (p_1(x_i) + p_2(x_i))^2 \leq 0$, so that $p_1(x_i) = p_2(x_i)$.
    Since $p_1(0) = p_2(0)$ as well, $p_1$ and $p_2$ agree at $d+2$ points; but since they are both polynomials of degree $d$, we must have $p_1 = p_2$.
    Thus we have reached a contradiction.

\end{proof}

\begin{lemma} \label{lem:minmax-chebyshev-value}
    Let $q(n) = \frac{\sqrt{n} + 1}{\sqrt{n} - 1}$.
    Then
    \begin{equation} \label{eq:minmax-chebyshev}
        \min_{\substack{p \in P_d \\ p(0) = 1}} \max_{x \in [1,n]} |p(x)| = \frac{2}{q(n)^d + q(n)^{-d}}. 
    \end{equation}
    Moreover, the polynomial $\hat T_d$ which achieves the minimax value is unique. 
\end{lemma}

\begin{proof}
    Let $T_d(x)$ be the Chebyshev polynomial of order $d$.
    Consider the affine transformation $m(x) = \frac{2x - (n+1)}{n-1}$ which maps from the interval $[1,n]$ to $[-1,1]$, and take the degree $d$ polynomial $\hat T_d(x) = \frac{T_d(m(x))}{T_d(m(0))}$.
    By definition, $\hat T_d(0) = 1$. 
    \maryt{Write down the zeros of the OG and mapped Chebyshev polynomial here:}
    Moreover, $\hat T_d$ satisfies the equioscillation condition of Lemma~\ref{lem:continuous-equioscillate} (it has $d+1$ points $x_i$ in $[1,n]$ where $\hat T_d (x_i) = \max_{x \in [1,n]} T_d(x)$), so it is the unique polynomial that achieves the minimax value.
    This minimax value is $\frac{1}{\hat T_d(m(0))} = \frac{1}{T_d\left( -\frac{n+1}{n-1}\right)}$; simplifying according to the identity $T_d(x) = \cosh (d \arcosh{x})$ for $x \geq 1$, the result follows. 
\end{proof}

\begin{lemma} \label{lem:chebyshev-lower-bound}
    \[ \min_{\substack{p \in P_d \\ p(0) = 1}} \max_{t \in \{1,\dots,n\}} |p(t)| \geq \left( 1 - \frac{d^2}{n-1}\right) \min_{\substack{p \in P_d \\ p(0) = 1}} \max_{x \in [1,n]} |p(x)|. \]
\end{lemma}

\begin{proof}
    We first map grid points $t \in [n]$ to a finite set $J \subseteq [-1, 1]$ of evenly-spaced points with the mapping $m(t) = \frac{2t - (n+1)}{n-1}$. 
    The distance between any two grid points $t_1, t_2 \in J$ is $\delta = \frac{2}{n-1}$. 
    Pick any $x \in [-1,1]$ and let $t_j$ be the closest grid point to $x$, so that they are at most distance $\delta/2$ from each other. 
    Then for any polynomial $p$ of degree at most $d$, 
    \begin{align*}
        |p(x)| - |p(t_j)| &\leq |p(x) - p(t_j)| \\
        &\leq \frac{\delta}{2} \max_{x \in [-1,1]} |p'(x)| \; &\text{(by the mean value theorem)}\\
        &\leq \frac{\delta}{2} d^2 \max_{x \in [-1,1]} |p(x)| \; &\text{(by the Markov brothers' inequality)} \\
        &= \frac{d^2}{n-1} \max_{x \in [-1,1]} |p(x)|.
    \end{align*}

    Rearranging, we have for any $x \in [1,1]$ that $|p(x)| \leq \frac{d^2}{n-1}  \max_{x \in [-1,1]} |p(t)| + |p(t_j)|$, so
    \begin{align*}
        \max_{x \in [-1,1]} |p(x)| &\leq \left(\frac{d^2}{n-1}\right) \max_{x \in [-1,1]} |p(x)| + \max_{t \in J} |p(t)| \\
        \left(1 - \frac{d^2}{n-1}\right) \max_{x \in [-1,1]} |p(x)| &\leq \max_{t \in J} |p(t)|.
    \end{align*}
    The statement follows by applying the affine mapping $m^{-1}(x)$ to the intervals and noting the inequality holds for the minimizing polynomial in the set $\{p \in P_d: p(0) = 1\}$.
\end{proof}

\begin{proof}[Proof of Theorem~\ref{thm:agent-order-error}]
    Note first that \[\min_{\substack{p \in P_d \\ p(0) = 1}} \max_{t \in [n]} |p(t)| \leq \min_{\substack{p \in P_d \\ p(0) = 1}} \max_{x \in [1,n]} |p(x)|,\] and by increasing monotonicity of the function $f(x) = x^2$ on $[1, \infty)$,
    \[\min_{\substack{p \in P_d \\ p(0) = 1}} \max_{t \in [n]} p(t)^2 = \left(\min_{\substack{p \in P_d \\ p(0) = 1}} \max_{t \in [n]} |p(t)|\right)^2.\]
    Combining Lemmas~\ref{lem:intersection-minimax-problem}, ~\ref{lem:minmax-chebyshev-value}, then,
    \[\min_{A \in \Lorder} \error(A) = \min_{\substack{p \in P_d \\ p(0) = 1}} \max_{t \in [n]} p(t)^2 \leq \left(\min_{\substack{p \in P_d \\ p(0) = 1}} \max_{x \in [1,n]} |p(x)|\right)^2 = \frac{4}{\left(q(n)^d + q(n)^{-d}\right)^2}\]
    for $q(n) = \frac{\sqrt{n} + 1}{\sqrt{n} - 1}$, and similarly by Lemma~\ref{lem:chebyshev-lower-bound},
    \[\min_{A \in \Lorder} \error(A) \geq \frac{4\left(1 - \frac{d^2}{n}\right)^2}{\left(q(n)^d + q(n)^{-d}\right)^2}.\]
    Now, note that $4 \left(q(n)^d + q(n)^{-d}\right)^{-2} \leq 4 q(n)^{2d} = 4 \left(1 - \frac{2}{\sqrt{n}+1}\right)^{2d} \leq 4e^{\frac{-4d}{\sqrt{n}+1}}$. 
    The first statement immediately follows. 

    \maryt{I'm being a little lazy with asymptotics here... clean up if time}
    Meanwhile, let $d = \Theta(\sqrt{n})$, so for some constants $c_1, c_2 > 0$ and $n_0$, if $n \geq n_0$ then $c_1 \sqrt{n} \leq d \leq c_2 \sqrt{n}$.
    Using the standard logarithmic Taylor expansion, we have for $u \coloneqq 1/\sqrt{n}$ that
    $\ln q(n) = \ln(1+u) - \ln(1-u) = 2 \left(u + \frac{u^3}{3} + O(u^5)\right) = \frac{2}{\sqrt{n}} + O(n^{-3/2}).$
    Then for any $d$ in the interval $[c_1 \sqrt{n}, c_2 \sqrt{n}]$, we have
    $q(n)^d + q(n)^{-d} = e^{d \ln q(n)} + e^{-d \ln q(n)} \to e^{2c} + e^{-2c} + o(1)$
    for some $c \in [c_1, c_2]$ as $n \to \infty$. Thus
    \[\lim_{n \to \infty} 4 \left(q(n)^d + q(n)^{-d}\right)^{-2} = 4(e^{2c} + e^{-2c})^{-2} \in (0,1).\]
    Similarly, $1 - \frac{d^2}{n-1} \in [1 - c_2^2, 1 - c_1^2] \subseteq (0,1)$ for all sufficiently large $n$. Multiplying by the previous limit gives
    \[
    \lim_{n \to \infty} \frac{4 (1 - \frac{d^2}{n-1})^2}{(q(n)^d + q(n)^{-d})^2} = C \in (0,1),
    \]
    where $C$ depends on $c_1, c_2$.

    Finally, let $d = o(\sqrt{n})$.
    We will show that $4 \left(q(n)^d + q(n)^{-d}\right)^{-2} = 1 - \frac{4d^2}{n} + o(\frac{d^2}{n}) = 1 - \Theta(\frac{d^2}{n})$, which immediately gives us the upper bound on error. 
    It follows by Lemma~\ref{lem:chebyshev-lower-bound} that $\min_{R \in \Lorder} \error(R) \geq 4 (1 - \frac{d^2}{n-1})^2 \left(q(n)^d + q(n)^{-d}\right)^{-2} = 1 - \Theta(\frac{d^2}{n})$.

    Let $t \coloneqq d \ln q(n)$, and note that $4 \left(q(n)^d + q(n)^{-d}\right)^{-2} = \cosh^{-2}(t).$
    Using the Taylor expansion for the logarithm with $u \coloneqq 1/\sqrt{n}$, we have
    $\ln q(n) = \ln(1+u) - \ln(1-u) = 2 \left(u + \frac{u^3}{3} + O(u^5)\right) = \frac{2}{\sqrt{n}} + O(n^{-3/2}).$
    Thus
    $t = d \ln q(n) = \frac{2d}{\sqrt{n}} + o\Big(\frac{d}{\sqrt{n}}\Big)$ since $d = o(\sqrt{n}).$
    Next, using the Taylor expansion of $\cosh t$, we have
    $\cosh t = 1 + \frac{t^2}{2} + R(t)$ for $R(t) \geq 0$, $R(t) = O(t^4),$ so that
    $4 \left(q(n)^d + q(n)^{-d}\right)^{-2} = \cosh^{-2}(t) = 1 - t^2 + o(t^2).$
    Substituting $t^2 = (d \ln q(n))^2 = \frac{4d^2}{n} + o\Big(\frac{d^2}{n}\Big)$ gives the asymptotic expression
    \[
    \frac{4}{(q(n)^d + q(n)^{-d})^2} = 1 - \frac{4 d^2}{n} + o\Big(\frac{d^2}{n}\Big).
    \]

\end{proof}

\begin{lemma} \label{lem:discrete-equioscillate}
    \[p^* \in \arg\min_{\substack{p \in P_d \\ p(0) = 1}} \max_{t \in [n]} |p(t)|\] if and only if $p^*$ equioscillates at $d+1$ integers in $[1,n]$ (i.e. there exist $d + 1$ integers $1 \leq x_0 \leq \ldots \leq x_d \leq n$ such that $|p^*(x_i)| = \sigma(-1)^i \max_{x \in [1,n]} p^*(x)$, with $\sigma$ either $-1$ or $+1$). 
\end{lemma}

\begin{proof}
    The proof of this lemma exactly follows the equioscillation theorem shown in Lemma~\ref{lem:minmax-chebyshev-value}, but with the maximum in the problem taken over a finite set of points. 
\end{proof}

\begin{lemma} \label{lem:exact-d-bounds}
    \[ \min_{\substack{p \in P_d \\ p(0) = 1}} \max_{t \in \{1,\dots,n\}} |p(t)| = \min_{\substack{p \in P_d \\ p(0) = 1}} \max_{x \in [1,n]} |p(x)| = \frac{2}{q(n)^d + q(n)^{-d}}\]
    if and only if one of the cases hold: $d = 1$, $d = 2$ and $n$ is odd, or $d = 3$ and $n \equiv 1 \pmod 4$.
\end{lemma}

\begin{proof}
    $T_d$ achieves its maxima $\max_{x \in [-1,1]} T_d(x)$ at points $x_k = \cos (\frac{k \pi}{d})$ for $0 \leq k \leq d$, so that the extrema of $\hat T_d(x)$ occur at $\hat x_k = \frac{1}{2} \left((n-1) \cos (\frac{k \pi}{d}) + (n+1)\right)$.
    Moreover, $\hat T_d$ is the unique polynomial that achieves the optimal minimax value for the continuous minimax problem by Lemma~\ref{lem:minmax-chebyshev-value}.
    By Lemma~\ref{lem:discrete-equioscillate}, then, $\hat T_d$ achieves the exact minimax error for our discrete problem if and only if each $\hat x_k$ is an integer (by construction of $\hat T_d$, equioscillation is already achieved).
    By Niven's theorem, the only values $0 \leq \frac{k}{d} \leq 1$ such that $\cos(\frac{k\pi}{d})$ are rational are $\frac{k}{d} = 0, \frac{1}{3}, \frac{1}{2}, \frac{2}{3},$ or $1$.
    We can then break into cases over $d$ and $k$ to determine when $\hat x_k$ is an integer.
\begin{itemize}
  \item $k = 0$: then $\hat x_0 = n$, so the statement is true. 
  \item $d$ is even, and $k = \tfrac{d}{2}$: then $\cos \left(\tfrac{k\pi}{d}\right) = 0$, so $\hat x_k = \frac{n+1}{2},$ which is an integer if and only if $n$ is odd.

  \item $3 \mid d,$ and $k = \tfrac{d}{3}$: $\cos\left(\tfrac{k\pi}{d}\right) = \tfrac{1}{2}$, so
  \[
  \hat x_k = \frac{(n-1)(\tfrac{1}{2}) + (n+1)}{2} = \frac{3n + 1}{4},
  \]
  which is an integer if and only if $n \equiv 1 \pmod{4}$.

  \item $3 \mid d,$ and $k = \tfrac{2d}{3}$: $\cos\left(\tfrac{k\pi}{d}\right) = -\tfrac{1}{2}$, so
  \[
  \hat x_k = \frac{(n-1)(-\tfrac{1}{2}) + (n+1)}{2} = \frac{n + 3}{4},
  \]
  which is an integer if and only if $n \equiv 1 \pmod{4}$.
  \item $k = d$: then $\hat x_d = 1$, so the statement is true. 
\end{itemize}

For all other $k$ and $d$ pairs, $\cos(k\pi/d)$ is irrational, and thus $\hat x_k$ cannot be an integer.
Now, if $d = 1$, then $\hat x_0$ and $\hat x_1$ are both integers, so the statement holds.
If $d = 2$, $\hat x_1$ is an integer if and only if $n$ is odd. 
If $d = 3$, $\hat x_1$ and $\hat x_2$ are integers if and only if $n \equiv 1 \pmod{4}$.
For $d > 3$, $\hat x_{d-1}$ does not fall into any of the cases. 
\end{proof}

\begin{proof}[Proof of Proposition~\ref{prop:exact-small-d}]
    We have
    \begin{align*}
        \min_{A \in \Lorder} \error(A) &= \min_{\substack{p \in P_d \\ p(0) = 1}} \max_{t \in [n]} p(t)^2 & \text{(by Lemma~\ref{lem:intersection-minimax-problem})} \\
        &= \left(\min_{\substack{p \in P_d \\ p(0) = 1}} \max_{t \in [n]} |p(t)|\right)^2 & \text{(by increasing monotonicity of $f(x) = x^2$)}\\
        &= \frac{4}{\left(q(n)^d + q(n)^{-d}\right)^2},
    \end{align*}
    where the last line follows by Lemma~\ref{lem:exact-d-bounds} if and only if $d = 1$, $d = 2$ and $n$ is odd, or $d = 3$ and $n \equiv 1 \pmod 4$.
\end{proof}

%%% Local Variables:
%%% mode: latex
%%% TeX-master: "stoc26-draft"
%%% End:

\end{document}